%% file: main.tex
\documentclass[runningheads]{llncs}

\usepackage{balance}
\usepackage{multirow}

\usepackage{hyperref}

\usepackage[linesnumbered,ruled,vlined]{algorithm2e} 

\usepackage{enumitem}
\usepackage{graphicx} 
\usepackage{graphics}
\usepackage{latexsym}
\usepackage{pgfplots}
\usepackage{amssymb}
\usepackage{tikz}
\usetikzlibrary{patterns,arrows}
\usepackage{booktabs}
\usepackage{pgfplotstable}
\usepgfplotslibrary{fillbetween} 
\usepackage{subcaption} 
\usepgfplotslibrary{statistics}
\pgfplotsset{compat=newest,
}

\usepgfplotslibrary{groupplots}

\usetikzlibrary{spy}

\usepackage{rotating} 

\usepackage{amsmath,amsfonts}
\usepackage{algorithmic}
\usepackage{graphicx}
\usepackage{textcomp}

\usepackage[misc]{ifsym}
\usepackage{fancyhdr}

\begin{document}
%
\title{Cross-Model Conjunctive Queries over Relation and Tree-structured Data (Extended)}
%

\author{Yuxing Chen \and
Jiaheng Lu~{\scriptsize}}
%
\institute{University of Helsinki \\
\email{first.last@helsinki.fi}}

\maketitle              
\begin{abstract}
Conjunctive queries are the most basic and central class of database queries. With the continued growth of demands to manage and process the massive volume of different types of data, there is little research to study the conjunctive queries between relation and tree data.
In this paper, we study of Cross-Model Conjunctive Queries (CMCQs) over relation and tree-structured data (XML and JSON). 
To efficiently process CMCQs with bounded intermediate results, we first encode tree nodes with position information. With tree node original label values and encoded position values, it allows our proposed
algorithm \textit{CMJoin} to join relations and tree data simultaneously, avoiding massive intermediate results. \textit{CMJoin} achieves worst-case optimality in terms of the total result of label values and encoded position values. 
Experimental results demonstrate the efficiency and scalability of the proposed techniques to answer a CMCQ in terms of running time and intermediate result size. 

\end{abstract}
\thispagestyle{empty}
\input{section/introduction}

\input{section/preliminary}

\input{section/size_bound}

\input{section/approach}

\input{section/evaluations}
\input{section/discussion}

\input{section/relatedwork}

\input{section/conclusion}
\input{section/acknowledgments}
\input{section/appendix}

%
%
%
\bibliographystyle{plain}
\bibliography{reference}

\end{document}

%% file: section/introduction.tex
\section{Introduction}\label{sec:introduction}

Conjunctive queries are the most fundamental and widely used database queries   \cite{DBLP:conf/focs/AtseriasGM08}. 
They correspond to \texttt{project}-\texttt{select}-\texttt{join}  queries in the relational algebra. They also correspond to non-recursive datalog rules \cite{DBLP:conf/pods/ChaudhuriV92}
\begin{equation}\label{equ:conjunctive_join}
R_0(u_0) \leftarrow R_1(u_1) \wedge R_2(u_2) \wedge \ldots \wedge R_n(u_n),
\end{equation}
where $R_i$ is a relation name of the underlying database, $R_0$ is the output relation, and each argument $u_i$ is a list of $|u_i|$ variables, where $|u_i|$ is the arity of the corresponding relation. The same variable can occur
multiple times in one or more argument lists.

It turns out that traditional database engines are not optimal to answer  conjunctive queries, as all pair-join engines may produce unnecessary intermediate results on many join queries~\cite{DBLP:journals/corr/abs-1203-1952}. For example, consider a typical triangle conjunctive query $R_0(a,b,c) \leftarrow R_1(a,b) \wedge R_2(b,c) \wedge R_3(a,c)$, where the size of input relations $|R_1|$= $|R_2|$ =$|R_3|$ = $N$.  The worst-case size bound of the output table $|R_0|$ yields $\mathcal{O}(N^{\frac{3}{2}})$. But any pairwise relational algebra plan takes at least $\Omega(N^{2})$, which is asymptotically worse than the optimal engines. To solve this problem, recent algorithms (e.g. \textit{NPRR} \cite{DBLP:journals/corr/abs-1203-1952},  \textit{LeapFrog}  \cite{veldhuizen2012leapfrog}, \textit{Joen}  \cite{ciucanu2015worst}) were discovered to achieve the optimal asymptotic bound for conjunctive queries. 

Conjunctive queries over trees have recently attracted attention \cite{DBLP:journals/jacm/GottlobKS06}, as trees are a clean abstraction of HTML, XML, JSON, and LDAP. The tree structures in conjunctive queries are represented using node label relations and axis relations such as \textit{Child} and \textit{Descendant}. For example, the \textit{XPath} query $A[B]//C$ is equivalent to the conjunctive query:  

\begin{equation}
\begin{split}
R(z) &\leftarrow Label(x,``A") \wedge Child(x,y) \wedge Label(y,``B")  \\
           & ~~~~ \wedge Descendant(x,z) \wedge Label(z,``C").
\end{split}
\end{equation}

Conjunctive queries with trees have been studied extensively. For example, see \cite{DBLP:journals/jacm/GottlobKS06} on their complexity, \cite{BENEDIKT20053} on their expressive power, and \cite{Bjorklund:2007:CQC:1783534.1783542,DBLP:journals/jacm/GottlobKS06} on the satisfiability problem. 
While conjunctive queries with relations or trees have been studied separately in the literature, to the best of our knowledge, there is no existing work to combine them together to study a hybrid conjunctive query. This paper fills this gap to embark on the study of  a Cross-Model Conjunctive Query (CMCQ) over both relations and trees. This  problem emerges in modern data management and analysis, which often demands a hybrid evaluation with data organized in different formats and models, as illustrated in the following. 

\begin{figure}[htp!]
     \centering
    \small
     \includegraphics[width=0.8\linewidth]{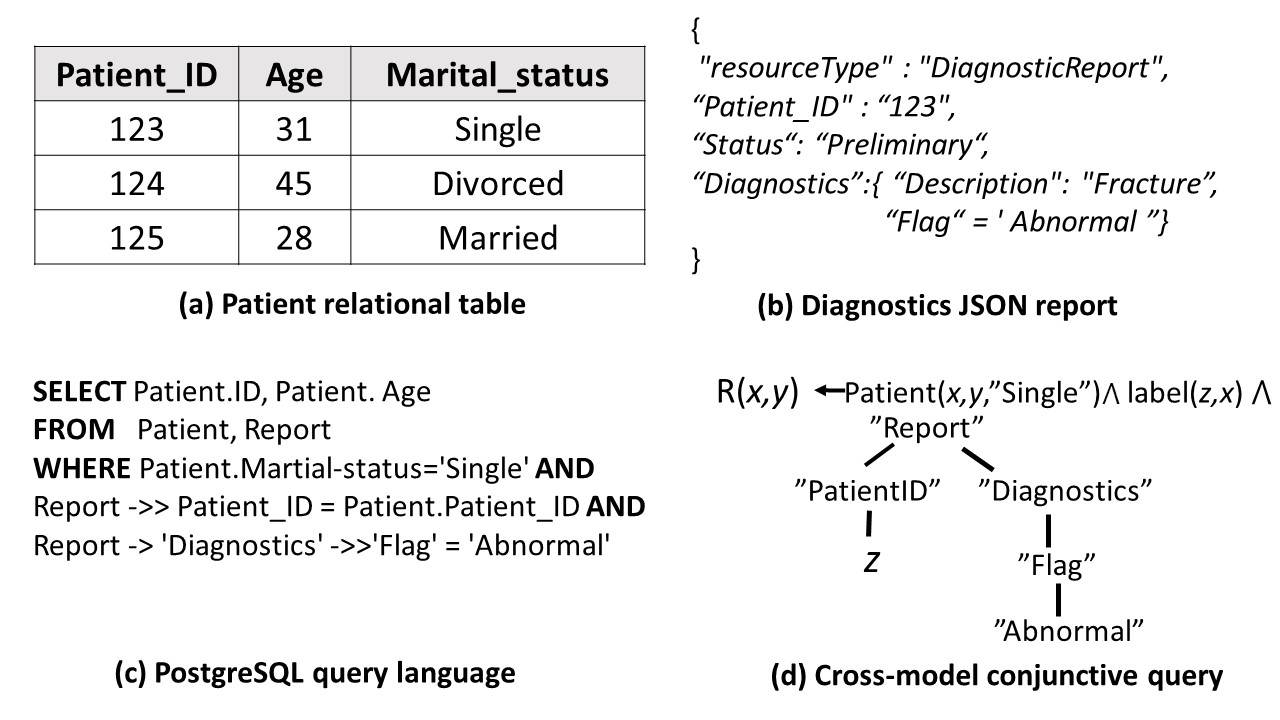}
     \caption{Illustration for the motivation and example of cross-model conjunctive query.}
     \label{fig:example_motivation_cross_model_query}
\end{figure}


\begin{example}\label{example:intro_cross_model_example} Suppose we want to perform  data analysis  synergistically for heterogeneous medical records where the patient information is stored in a relational table and the diagnostic reports are formatted with JSON documents (See \autoref{fig:example_motivation_cross_model_query}). Assume that one query is to
find the patient who is single and has an  ``\textsf{abnormal}'' flag in the diagnostic report.   \autoref{fig:example_motivation_cross_model_query}(c) illustrates the query language based on  PostgreSQL database to perform the cross-model join between the relational table and the JSON document.  This query can be naturally represented with a cross-model conjunctive query in \autoref{fig:example_motivation_cross_model_query}(d). 
\end{example}

\begin{figure}\centering 
\input{example/intro_cmcq.tex} \caption{An example of a CMCQ.} \label{fig:cross_model_query} \end{figure}
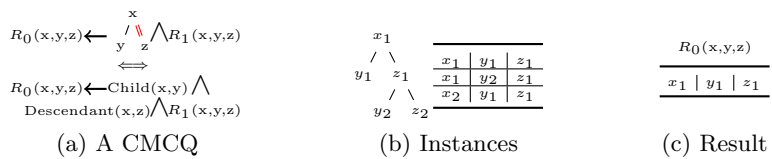

This paper embarks on the study of  a Cross-Model Conjunctive Query (CMC Q) over both relations and trees. \autoref{fig:cross_model_query} depicts a CMCQ. CMCQs emerge in modern data management and analysis, which often demands a hybrid evaluation with data organized in different formats and models, e.g. data lake~\cite{DBLP:conf/sigmod/HaiGQ16}, multi-model databases \cite{Lu-2019-MDN-3341324-3323214}, polystores~\cite{DBLP:journals/sigmod/DugganESBHKMMMZ15}, and computational linguistics \cite{xue2005penn_treebank}.
The detailed application scenarios of CMCQs can be exemplified  as follows:

\smallskip
\noindent \textbf{Data integration in data lake} \indent Data which resides in the data lake \cite{DBLP:conf/sigmod/HaiGQ16} may include highly structured data stored in SQL databases or data warehouses, and nesting or multiple values data in Parquet or JSON documents. Cross-model conjunctive queries on trees and relations can be used to integrate structured data from relational databases and semi-structured  data in open data formats (e.g., through the PartiQL query language in Amazon Data lake \cite{journals/debu/CaiGGNPPP18}).

 \smallskip
\noindent \textbf{Cross-model query processing} \indent In the scenarios of multi-model databases \cite{Lu-2019-MDN-3341324-3323214} and polystores~\cite{DBLP:journals/sigmod/DugganESBHKMMMZ15},  query evaluation often involves data formatted with different models. The join problems between the structured and semi-structured data can boil down to evaluating a CMCQ 
over relations and trees. 

 \smallskip
\noindent \textbf{Queries in computational linguistics} \indent A further area in which CMCQs are employed is computational linguistics, where one needs to search in, or check properties of large corpora of parsed natural language. Corpora such as Penn Treebank \cite{xue2005penn_treebank} are unranked trees labeled with the phrase structure of parsed texts. A conjunctive query on trees and relations can find sentences to satisfy specific semantic structure by  relations (e.g. hyponym relations \cite{journals/tkde/WeiLMZZF14}) and corpora trees \cite{zhou-zhao-2019-head}.

The number of applications that we have hinted at above motivates the study of CMCQs, and the main contributions of this paper are as follows:

\begin{enumerate}[leftmargin=*]

\item This paper embarks on the study of the cross-model conjunctive query (CMC Q) and formally defines the problem of CMCQ processing, which integrates both relational conjunctive query and tree conjunctive pattern together.

\item We propose \textit{CMJoin}-algorithm to process relations and encoded tree data efficiently. \textit{CMJoin} produces worst-case optimal join result in terms of the label values as well as the encoded information values. In some cases, \textit{CMJoin} is worst-case optimal join in the absence of encoded information.

\item Experiments on real-life and benchmark datasets show the effectiveness and efficiency of the algorithm in terms of running time and intermediate result size. 

\end{enumerate}

The remainder of the paper is organized as follows.  In \autoref{sec:preliminary} we provide preliminaries of approaches.  We then extend the worst-case optimal algorithm for CMCQs in \autoref{sec:approach}. We evaluate our approaches empirically in \autoref{sec:evaluation}.  We review related works in \autoref{sec:related_work}. \autoref{sec:conclusion} concludes the paper. Note that this is an extended version of previous work \cite{DBLP:conf/sigmod/Chen18,chen2021performance,DBLP:conf/dasfaa/ChenULLD22}

%% file: example/intro_cmcq.tex
\begin{subfigure}[t]{0.37\textwidth}
        \centering
\begin{tikzpicture}[font=\tiny,
	level distance=0.45cm,
  level 1/.style={sibling distance=.32cm},
  level 2/.style={sibling distance=1cm},
]
\useasboundingbox [fill=gray!0]  (0,0) rectangle (30.5mm,16.5mm); 

\node[black] at (4.5mm,12mm)  {$R_0$(x,y,z)};

\draw  [<-, thick] (9.5mm,12mm) -- (12.5mm,12mm);

\node[black] at (16mm,8mm)  { $\iff$ };

\node[black] at (16mm,15mm)  {x}
    child {node {y}}
    child {node[black] {z}[style = {double,red}]};

\node[black] at (19.5mm,12.5mm)  {\large $\wedge$ \normalsize};

\node[black] at (25.5mm,12mm)  {$R_1$(x,y,z)};

\node[black] at (18mm,5mm)  {Child(x,y)};
\node[black] at (25mm,5.5mm)  {\large $\wedge$ \normalsize};
\node[black] at (10mm,2mm)  {Descendant(x,z)};
\node[black] at (4.5mm,5mm)  {$R_0$(x,y,z)};
\draw  [<-, thick] (9.5mm,5mm) -- (12.5mm,5mm);
\node[black] at (19.5mm,2.5mm)  {\large $\wedge$ \normalsize};
\node[black] at (25.5mm,2mm)  {$R_1$(x,y,z)};

\end{tikzpicture}

        \vspace*{-1.5mm}
        \caption{A CMCQ}
        \label{fig:pre_cross_model_query}
\end{subfigure}%
\begin{subfigure}[t]{0.33\textwidth}
        \centering
\begin{tikzpicture}[font=\tiny,
	level distance=0.5cm,
  level 1/.style={sibling distance=0.5cm},
  level 2/.style={sibling distance=0.5cm},
]
\node[black] at (4mm,12.5mm)  {$x_1$}
    child {node {$y_1$}}
    child {node {$z_1$}
        child {node {$y_2$}}
        child {node {$z_2$}}};

     \node at (18mm,7.5mm) {\setlength{\tabcolsep}{3pt}\begin{tabular}{ c | c | c}
     \toprule
     $x_1$ & $y_1$ & $z_1$ \\ \hline
     $x_1$ & $y_2$ & $z_1$ \\ \hline
     $x_2$ & $y_1$ & $z_1$ \\ \bottomrule
\end{tabular}};

\end{tikzpicture}
    
        \vspace*{-1.5mm}
        \caption{Instances}
        \label{fig:pre_cross_model_instance}
\end{subfigure}%
\begin{subfigure}[t]{0.25\textwidth}
        \centering
\begin{tikzpicture}[font=\tiny,
	level distance=0.5cm,
  level 1/.style={sibling distance=0.5cm},
  level 2/.style={sibling distance=0.5cm},
]
  \useasboundingbox [fill=gray!0] (0,0) rectangle (17mm,16mm);
  \node[black] at (8.5mm,10.5mm)  {$R_0$(x,y,z)};
  
  \node at (8.5mm,6mm) {\setlength{\tabcolsep}{3pt}\begin{tabular}{ c | c | c}
     \toprule
     $x_1$ & $y_1$ & $z_1$ \\ \bottomrule
\end{tabular}};
\end{tikzpicture}

        \vspace*{-1.5mm}
        \caption{Result}
        \label{fig:pre_cross_model_result}
\end{subfigure}%








%% file: section/preliminary.tex
\section{Preliminary}\label{sec:problem_statement}
\label{sec:preliminary}


\noindent \textbf{Cross-model conjunctive query} \indent Let $\mathbb{R}$ be a database schema and $R_1,\ldots,$ $R_n$ be relation names in $\mathbb{R}$. A rule-based conjunctive query over $\mathbb{R}$ is an expression of the form  $R_0(u_0) \leftarrow R_1(u_1) \wedge R_2(u_2) \wedge \ldots \wedge R_n(u_n)$, where $n \geq 0$, $R_0$ is a relation not in $\mathbb{R}$. Let $u_0, u_1, \ldots, u_n$ be free tuples, i.e. they may be either variables or constants. Each variable occurring in $u_0$ must also occur at least once in $u_1, \ldots, u_n$. 

Let $T$ be a tree pattern with two binary axis relations: \texttt{Child}  and  \texttt{Descendant}. The axis relations \texttt{Child}  and  \texttt{Descendant} are defined in the normal way \cite{DBLP:journals/jacm/GottlobKS06}. In general,  a cross-model conjunctive query contains three components: (i) the relational expression $\tau_1 := \exists r_1,\ldots, r_k \colon R_1(u_1) \wedge R_2(u_2) \wedge \ldots \wedge R_n(u_n)$, where $r_1, \ldots, r_k$ are all the variables in the relations $R_1,\ldots, R_n$;
(ii) the tree  expression $\tau_2 := \exists t_1, \ldots, t_k \colon Child(v_1)  \wedge \ldots \wedge Descendant(v_n)$,
where $t_1,\ldots,t_k$ are all the node variables occurring in $v_i$, for $i \geq 1$ and each $v_i$ is a binary tuple $(t_{i_1},t_{i_2})$; and (iii) the cross-model label expression $\tau_3 := \exists  r_1, \ldots, r_k, t_1, \ldots ,t_k \colon label_1(t_{i_1},r_{j_1}) \wedge \ldots \wedge label_n(t_{i_n},r_{j_n})$, where $\Sigma$ denotes a labeling alphabet. Given any node $t \in T$,  $label(t,s)$ means that the label of the node $t$ is $s \in \Sigma$. The label relations bridge the expressions of relations and trees by the equivalence between the label values of the tree nodes and the values of relations.

By combining the three components together, we define a cross-model conjunctive query with the calculus of form $\left\{e_{1}, \ldots, e_{m} \mid \tau_1  \wedge  \tau_2 \wedge \tau_3 \right\} \label{eq:cmcq}$, where the variables $e_1,\ldots,e_m$ are the return elements which occur at least once in relations.
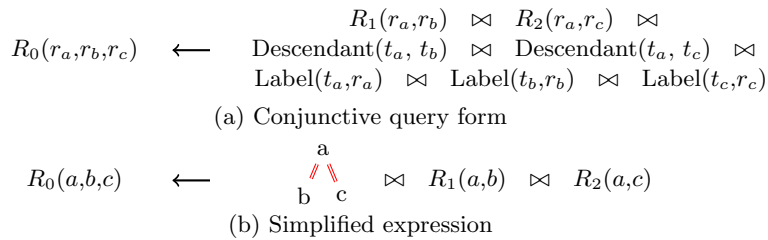
\begin{figure}[t]
     \centering
     \small
     \input{example/simplify_query}
     \caption{Complete and simplified expressions of CMCQ}
     \label{fig:simplified}
 \end{figure}
 \autoref{fig:conjunctive_query_form} shows an example of a cross-model conjunctive query, which includes two relations and one tree pattern. For the purpose of expression simplicity, we do not explicitly distinguish between the variable of trees (e.g. $t_a$) and that of relations (e.g. $r_a$), but simply write them with one symbol (i.e. $a$) if $label(t_a, r_a)$ holds.  We omit the label relation when it is clear from the context. 
\autoref{fig:simplify_query} shows a simplified representation of a query. 
 

\noindent \textbf{Revisiting relational size bound} \indent
We review  the size bound for the relational model, which Asterias, Grohe, and Marx (AGM) \cite{DBLP:conf/focs/AtseriasGM08} developed. 
The AGM bound is computed with linear programming (LP). Formally, given a relational schema $\mathbb{R}$, for every table $R \in \mathbb{R}$ let $A_{R}$ be the set of attributes of $R$ and $\mathcal{A} = \cup_{R} A_{R}$. Then the worst-case size bound is precisely the optimal solution for the following LP:

\begin{equation}\label{eq:lp_slution}
\begin{aligned}
 \quad  & \underset{}{\text{maximize}}
& &  \Sigma^\mathcal{A}_r x_r  \\
& \text{subject to}
& &  \Sigma^{A_{R}}_r x_r \leq 1  & & \text{ for all } R \in \mathbb{R},\\
&   
& &  0 \leq x_r \leq 1 & & \text{ for all } r\in \mathcal{A}.\\  
\end{aligned}
\end{equation}

Let $\rho$ denote the optimal solution of the above LP. Then the size bound of the query is $N^\rho$, where $N$ denotes the maximal size of each table. The AGM bound can be proved as a special case of the discrete version of the well-known Loomis-Whitney inequality \cite{loomis1949inequality} in geometry.  Interested readers may refer to the details of the proof in \cite{DBLP:conf/focs/AtseriasGM08}. We present these results informally and refer the readers to Ngo et al. \cite{DBLP:journals/sigmod/NgoRR13} for a complete survey.

For example, we consider a typical triangle conjunctive query $R_0(a,b,c) \leftarrow R_1(a,b) \wedge R_2(b,c) \wedge R_3(a,c)$ that we introduced in \autoref{sec:introduction}. Then the three LP inequalities corresponding to three relations include $x_a + x_b \leq 1$, $x_b + x_c \leq 1$, and $x_a + x_c \leq 1$. Therefore, the maximal value of $x_a + x_b + x_c$ is $3/2$, meaning that the size bound is $\mathcal{O}(N^{\frac{3}{2}})$. Interestingly, the similar case for CMCQ in \autoref{fig:alg_cross_model_query_example}, the query $Q$ = $a[b]/c \bowtie R(b,c)$ has also the size bound $\mathcal{O}(N^{\frac{3}{2}})$. 

%% file: example/simplify_query.tex
\begin{subfigure}[t]{0.80\textwidth}
        \centering
        
\begin{tikzpicture}[font=\small,
	level distance=0.6cm,
  level 1/.style={sibling distance=0.5cm},
  level 2/.style={sibling distance=0.5cm},
]
\useasboundingbox [fill=gray!0] (0,0) rectangle (80mm,9mm); 
\node[black] at (2mm,5mm)  {$R_0$($r_a$,$r_b$,$r_c$)};

\draw  [<-, thick] (15mm,5mm) -- (20mm,5mm);

\node[black] at (60mm,9mm)  {$R_1$($r_a$,$r_b$) $~\bowtie~$ $R_2$($r_a$,$r_c$)  $~\bowtie~$ };
\node[black] at (60mm,5mm)  { Descendant($t_a$, $t_b$) $~\bowtie~$  Descendant($t_a$, $t_c$)  $~\bowtie~$};
\node[black] at (60mm,1mm)  { Label($t_a$,$r_a$) $~\bowtie~$ Label($t_b$,$r_b$) $~\bowtie~$ Label($t_c$,$r_c$) };

\end{tikzpicture}
        \caption{Conjunctive query form}
        \label{fig:conjunctive_query_form}
\end{subfigure}%

\begin{subfigure}[t]{0.8\textwidth}
        \centering
        
\begin{tikzpicture}[font=\small,
	level distance=0.6cm,
  level 1/.style={sibling distance=0.5cm},
  level 2/.style={sibling distance=0.5cm},
]
\useasboundingbox [fill=gray!0] (0,0) rectangle (80mm,8mm); 

\node[black] at (2mm,2mm)  {$R_0$($a$,$b$,$c$)};

\draw  [<-, thick] (15mm,2mm) -- (20mm,2mm);

\node[black] at (60mm,2mm)  {$~\bowtie~$ $R_1$($a$,$b$) $~\bowtie~$ $R_2$($a$,$c$)};
\node[black] at (35mm,6mm)  {a}
    child {node[black] {b}[style = {double,red}]}
    child {node[black] {c}[style = {double,red}]};

\end{tikzpicture}
        \caption{Simplified expression}
        \label{fig:simplify_query}
\end{subfigure}%

        


        



%% file: section/size_bound.tex
\section{Worst-case Size bound}\label{sec:size_bound}

Given a CMCQ $Q$, our analysis will be carried out in two assumptions: (1) regarding the tree database, given any label $s \in \Sigma$, the number of nodes with the label $s$ is at most $N$; (2) regarding the relational database, the size of each table is also at most $N$. Based on these two assumptions,  we study how to find the maximum bound of the size of answers for $Q$ in the worst case. In this section, we start our investigation with two simple yet important special cases followed with a generalized algorithm and optimizations.



\subsection{Glance at two special cases}\label{sec:size_bound_multi_model}

\begin{figure}[t!]\centering
     \input{example/instance_tree}
     \caption{Tree conjunctive queries with (a) Descendant axes and (c) Child axes, and their worst-case instance tree (b) for query (a) and tree (d) for query (c).}
     \label{fig:example_tree_query}
     \end{figure}
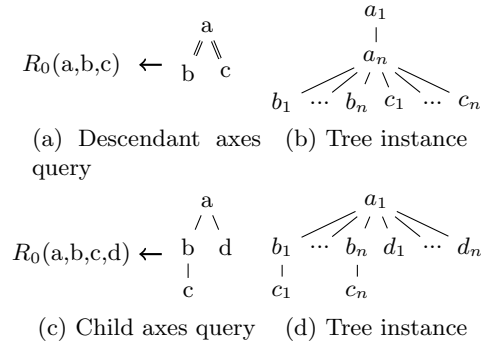

 A cross-model query $Q$ include two components: relation expressions and tree patterns. As mentioned in \autoref{sec:problem_statement}, the conditions of relational part can be captured through LP inequalities in the AGM bound. Now the key challenge is how to represent the tree structure with the inequalities. In fact, the most
important component of our algorithms is a method to appropriately define LP inequality constraints for the tree pattern. To this end, we start our journey with two special cases, where the tree pattern contains only child or descendant axes, which shed light on the computation of a general case later.  




\noindent \textbf{Only descendant axes} ~ When the tree pattern in $Q$ contains only descendant axes,  it is a lucky fluke, thus there is no need to add any extra inequalities beside the existing relational inequalities. For example, consider a query $Q$ in \autoref{fig:twig_with_ad} is a tree pattern. \autoref{fig:twig_with_ad_instance} shows an instance tree which realize the worst size bound. The result of $Q$ is $\mathcal{O}(N^3)$ from all combination result of three nodes. Therefore, the existence of tree patterns with only descendant axes does not require any extra inequalities.

\smallskip

\noindent \textbf{Only child axes} ~ When the tree pattern in $Q$ contains only child axes,   this case is different from the first one. For each  path in the tree pattern, we need to add one parent-child (PC)-path inequality constraint. This is because each child node can have only one  parent node to match the axis condition.

\begin{definition} [\textbf{PC-path inequality}]  Let $P$ be a PC path in the tree pattern $\mathbf{T}$ and $A_P$ denote the set of labels of nodes in $P$. The  PC-path inequality of $P$  is defined as $\Sigma^{A_P}_r x_r \leq 1$. \label{def:PCPath}
\end{definition}
 
The matching result of attributes in a inequality is $\mathcal{O}(N)$ in relation as they are in one table and in PC-path, as they have one-to-one parent-child relationship.
For example, see \ref{fig:twig_with_pc} for an example of query with only child axes. The pattern matching result of attribute ($a$,$b$,$c$) and ($a$,$d$) are both $\mathcal{O}(N)$. So, the PC-path inequalities are $x_a + x_b + x_c \leq 1$ and $x_a + x_d \leq 1$. The solution is 2 in this case, meaning that the result of $Q$ in \autoref{fig:twig_with_pc} is $\mathcal{O}(N^2)$. \autoref{fig:twig_with_pc_instance} shows one of the worst-case construction tree.

\smallskip
\smallskip

\noindent \textbf{Mixed child and descendant axes}  ~ Given a tree pattern with both child and descendants axis relations, one may wonder whether the relational inequalities and PC-path inequalities are sufficient to produce the correct bound. Unfortunately,  this situation is more complicated. The fact is that all relational and PC-path inequalities do not suffice to derive the correct bound,  as illustrated below. 

\begin{example}\label{example:size_bound4corss_model_query}
Consider the query in \autoref{fig:twig_with_confilct1} with only a tree pattern.  The corresponding PC-path inequalities are  $x_a+x_b \leq 1$,~$x_a+x_c \leq 1$,~and $x_d \leq 1$.  Thus, the maximum value of $x_a+x_b+x_c+x_d$ is 3 when $x_b$=$x_c$=$x_d$=1 and $x_a$=0 (one of the possible solution). However, it is infeasible to construct a tree instance with $\Theta(N^3)$ to match the result. In fact, the tight upper bound is only $\mathcal{O}(N^2)$. \autoref{fig:twig_with_confilct_instance1} and \autoref{fig:twig_with_confilct_instance2} show two instances of trees in the worst case situation. 

In this case, when we obtain $\mathcal{O}(N^2)$ result for $b$ and $c$ in \autoref{fig:twig_with_confilct_instance1}, $c$ and $d$ can not yield $\mathcal{O}(N^2)$ result any more, meaning the $c$ and $d$ is no more equivalent to no constraint for descendant axe. And vice verse. So bound seems to be two alternatives: (i) $x_a+x_b \leq 1$,~$x_a+x_c+x_d \leq 1$, and (ii) $x_a+x_b+x_c \leq 1$,~$x_d \leq 1$, responding to two instance trees. They obtain the same size bound $\mathcal{O}(N^2)$. These two different alternatives are meaningful to compute the size bound with more complex case.
For example, give relation $R_1(b,c,d)$ (corresponding to $x_b+x_c+x_d \leq 1$), then with inequalities $(ii)$, it leads to the maximum bound $2$, while with the inequalities (i) it can obtain only $\frac{3}{2}$.  For another example with relation $R_3(b,d)$ and $R_4(a,c,d)$ (corresponding to $x_b+x_d \leq 1$ and $x_a+x_c+x_d \leq 1$), this time with inequalities (i) is a winner with maximum value $2$, comparing to $\frac{3}{2}$ for the inequalities $(ii)$.
\end{example}

 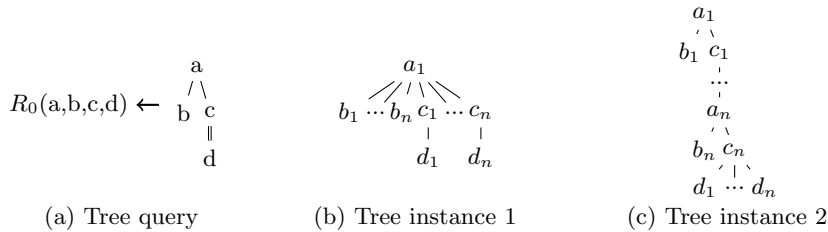
\begin{figure}[t!]
     \centering
    \small
     \input{example/conflict_twig_instance}
     \caption{Tree conjunctive query (a) and its alternative worst-case instance tree (b) and tree (c).}
     \label{fig:example_tree_query2}
 \end{figure}

Given a query with both ancestor and child axes, the above example hints at a possible approach to generate inequalities. That is, multiple options of LP problem settings need to be generated. In contrast to the AGM bound (with polynomial complexity),  the computation of the size bound for a CMCQ is in general $\mathcal{NP}$-hard with respect to query complexity (i.e. the complexity is measured in the size of the query).

\begin{theorem}[\textbf{NP-hardness}]\label{theo:np-hard} The query complexity of the worst-case bound evaluation of
databases for a cross-model conjunctive query is $\mathcal{NP}$-hard. 
\end{theorem}

The main idea of the proof is to polynomially reduce the 1-IN-3SAT problem~\cite{schaefer1978complexity} to our problem. See the Appendix in \autoref{sec:supplement_proof} for  the proof.

\smallskip

\noindent \textbf{Remark} ~ Note that this complexity is respect to the size of query. It should not be confused with the data complexity of query answering algorithm in \autoref{sec:approach} later, which has polynomial complexity with respect to data complexity. Considering a practical size of a query is limited, from the point of view of applications, the above theorem mainly makes it of theoretical interest. In contrast to relational conjunctive queries that are \textit{tractable} with respect to query complexity, this paper makes a contribution to demonstrate the theoretical complexity gap due to the occurrence of tree pattern in a conjunctive query.

\subsection{Recursive conversion and split}


This subsection develops a concrete algorithm to compute the worst-case bound. Here the high level idea is that we eliminate all descendant axis relations in the tree pattern recursively by two operations, called \textit{conversion} and \textit{split}, until the final tree pattern remains only child axis  which can be solved through LP solutions. Then we acquire the size bound by picking the maximum solution for all generated LP problems.

\begin{definition} [\textbf{Conversion and split operations}] Let T be a tree pattern and $\alpha$  be a descendant axis between $x$ and $y$ nodes in T.  Assume that $x$ is the parent of $y$.

\vspace*{-0.5mm}
\begin{enumerate}[leftmargin=*,label=-]
\item \textbf{Conversion}:  $T \mid \alpha$ denotes an operation to convert the descendant axis $\alpha$ to the child axis between $x$ and $y$.
\item \textbf{Split}:   $T \parallel \alpha$ denotes an operation to remove $\alpha$ from $T$ and thereby $T$ is split into two subtrees.  It is important to note that this split operation must be adjoined with one or multiple compensation inequalities defined below.
\end{enumerate}
\end{definition}


\begin{algorithm} \caption{Computing bound}\label{alg:computBound}\input{algorithm/Conversion.tex}\end{algorithm}

\begin{definition} [\textbf{Compensation inequalities}] With respective to the split operation in a tree pattern $T$, assume that node $x$ is split from node $y$. Let $P$ denote  the root-to-$x$ path. For each root-to-leaf path $P'$  that  does not contain $x$, let $\mathbf{A}$ denote  all labels  in $P$ and $P'$, then we generate a compensation inequality of $P'$:   $\Sigma x_r  \leq 1$  for all  labels $r \in \mathbf{A}$.
 \end{definition}
 
 To understand the reason of compensation inequalities, recall the tree pattern in 
 \autoref{example:size_bound4corss_model_query}. When node $c$ is split from node $d$, a compensation inequality is emitted: $x_a+x_b+x_c \leq 1$. This is because due to the split of node $c$ and $d$ in the tree pattern, in the worst-case tree instance,  each node $c$  should match each node $d$  (see \autoref{fig:example_tree_query2} c). Meanwhile, node $c$  must also have one parent node $a$  and this node $a$ must have at least one child $b$ node. Therefore,  $x_a+x_b+x_c \leq 1$ is necessary to capture those structural constraints.

\autoref{alg:computBound} illustrates the main steps to compute the bound. The input is a single tree pattern $T$ and relations $\mathbb{R}$. If there are multiple tree patterns, then we can easily merge them together by using a dummy root.  The output is the value of the worst-case size bound. The key idea of this algorithm is to generate all \textit{canonical suites} which can be converted to LP inequalities, as formally defined below.

\begin{definition} [\textbf{Suites and canonical suites}] A suite is a triple tuple ($\mathbb{R}$,$\mathbb{C}$,$\mathbb{T}$), where $\mathbb{R}$ denotes relations, $\mathbb{C}$  compensation inequalities and $\mathbb{T}$ tree patterns. In particular,  we say one suite is \textbf{canonical} if all tree patterns in $\mathbb{T}$ contain only child axes. A canonical suite can be directly converted  to a set of inequalities during the computation.
 \end{definition}

 
 Given any canonical suite ($\mathbb{R}$,$\mathbb{C}$,$\mathbb{T}$), the LP problem setting can be generated as follows:
 
 \begin{equation}\label{eq:lp_slution1}
\begin{aligned}
 \quad  & \underset{}{\text{maximize}}
& &  \Sigma^\mathcal{A}_r x_r  \\
& \text{subject to}
& &  \Sigma^{A_R}_r x_r \leq 1  & & for\ all\ R \in \mathbb{R} \\
&   
& &  \Sigma^{A_P}_r x_r \leq 1  & & for\ all \ P \in \mathbb{T}  \\
&   
& &  \Sigma^{A_C}_r x_r \leq 1  & & for\ all \ C \in \mathbb{C}  \\
&   
& &  0 \leq x_{r} \leq 1 & & for\ all\ r \in \mathcal{A}\\
\end{aligned}
\end{equation}

\begin{algorithm}[b!]
    \linespread{1.20}
    \caption{Find all canonical suites $CS(T,CI,R)$}
    \label{alg:simple_minors}
    \input{algorithm/MinorGeneration.tex}\end{algorithm}

\autoref{alg:simple_minors} illustrates the  procedure to generate all canonical suites. At a high level, the
main idea is to  traverse the query tree pattern $T$ in a top down fashion  to recursively eliminate all descendant axes through split and conversion operations.  Let us walk through the algorithm. If there is any descendant axis, then we pick a highest one $\alpha$,  to which there is no descendant axis in the path from root. Line 3 performs conversion operation and Line 4 adds the generated suites to $\mathbb{C}$.  Then Line 5 performs the conversion operation and Line 6-8 recursively call the functions to process subtrees which are generated from split operation.   When there is no descendant axis, Line 11 returns the canonical suite.

\begin{figure*}[t!]
     \centering
    \small
     \includegraphics[width=0.85\linewidth]{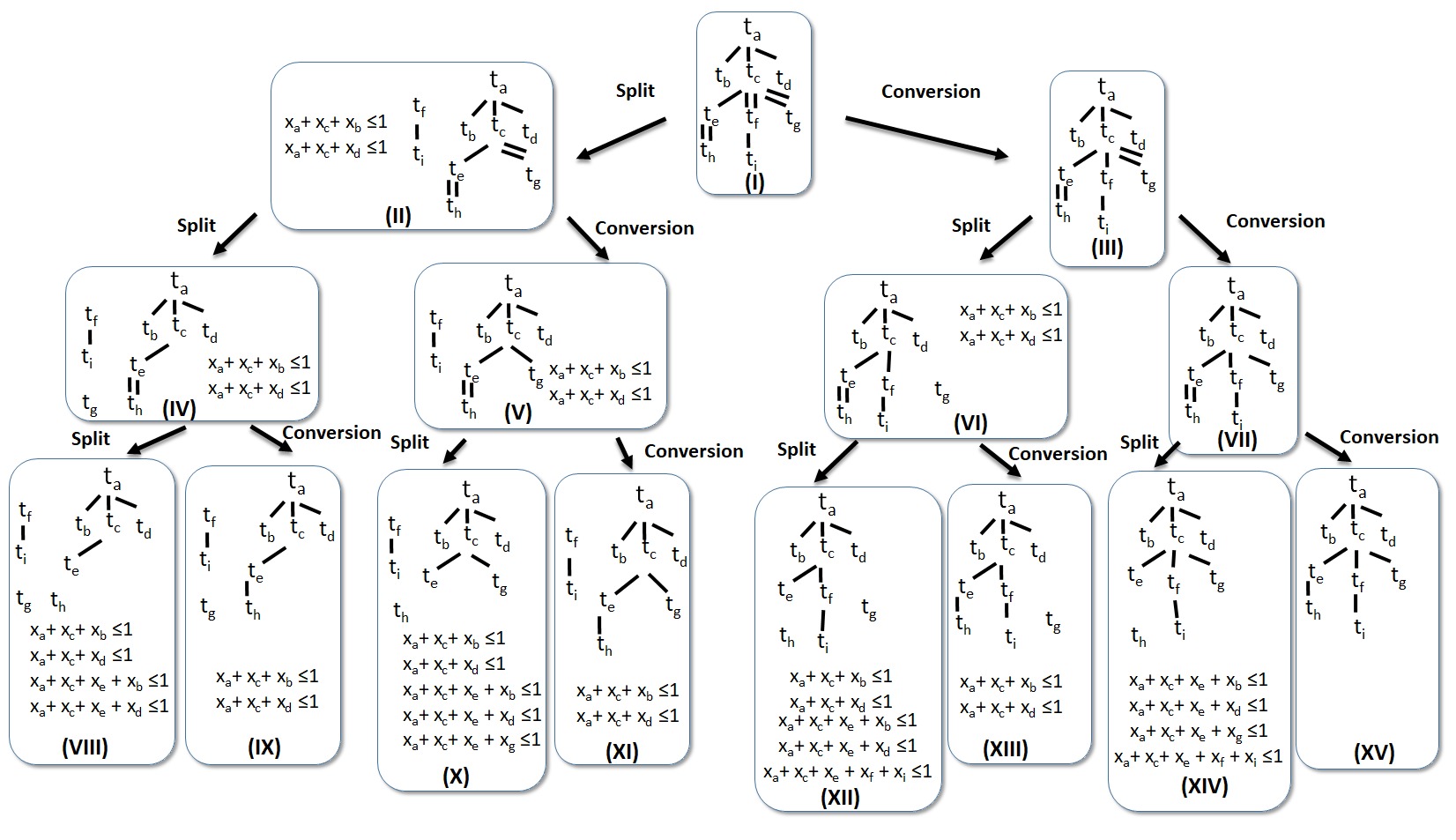}
     \vspace*{1mm}
     \caption{The tree query example for illustrating conversion and split.}
     \label{fig:new_example_mutation}
 \end{figure*}

\begin{example}\label{example:compute_size_bound_brute_force}
\autoref{fig:new_example_mutation}  depicts the detailed procedure to process a cross-model query, where the tree query is shown in \autoref{fig:new_example_mutation}(I) and the two relations are  $R_1(a,c,e,h)$ and $R_2(h,g,i)$. Since
$T$ has three  descendant axes, Algorithm 2 generates eight ($2^3$) sets of canonical suites. Then those canonical suites are converted to LP inequalities in Algorithm 1. By solving all LP problems, the maximal value of bound is $\mathcal{O}(N^{5})$, where $x_b = x_d = x_e =x_f = x_g = 1$.  In particular, the max bound of each suite is shown as follows: suite (VIII) = 4, suite (IX) = 5, suite (X) = 4, suite (XI) = 5,
suite (XII) = 4,
suite (XIII) = 5,
suite (XIV) = 4, and
suite (XV) = 5,
\end{example}

\begin{theorem}
Our algorithm finds the correct size bound. \label{them:correctOfAlgorithm}
\end{theorem}

The proof can be found in the Appendix.

\subsection{Two optimization rules}

Given a tree pattern with $n$ descendant axes,  the above algorithm  needs to generate all $2^n$ canonical suites to find the maximum bound. Although in the worst case,  this  complexity nature cannot  be alleviated if $NP \neq P$, this subsection  proposes two rules to significantly reduce the search space in most cases to improve the  efficiency.

\smallskip
\noindent \textbf{Optimization 1}:  Given any single tree pattern $T$, assume that there is already a conversion operation performed on $T$, then we observe that any following split operations can be safely canceled without affecting the correctness of the  computation. For example, recall \autoref{fig:new_example_mutation}. We can safely stop the computation for all conversion operations after the split. That is, suites (X), (XII), (XIII) and (XIV) can be canceled. Note that although suite (XIII)  generates the maximal value 5,  this value can be also provided from other suites (e.g. from suites (IX) and (XI)).

\begin{lemma}
Given a single tree pattern $T$, the conversion operations in $T$ after the spit are avoidable without affecting the final result of worst-case bound. \label{lem:split} 
\end{lemma}

The proof of \autoref{lem:split} can be found in Appendix. Intuitively, this result holds because, to construct the worst case tree instance,  a conversion operation requires the horizontal expansion of tree nodes, while a split operation demands the vertical expansion. The key observation is that the vertical expansion after the horizontal expansion cannot produce a new worst-case bound. Therefore, there is no need to carry out the conversion operation after the spit.  Further,  it is worthy to note that the above optimization can be applied only if the split and conversion operations are performed within the same tree. If the tree is split into two separated subtrees, then the split operation in the other subtree cannot be canceled.


\smallskip
 
\noindent \textbf{Optimization 2}:  Given any descendant axis $\alpha$, if there is no more descendant axis under $\alpha$, then we call $\alpha$ a leaf descendant axis. We observe that, if the condition in \autoref{lemma:optimization2} is satisfied, then the split operation for $\alpha$ can be safely canceled.

\begin{lemma} Given a leaf descendant axis $\alpha$ between node $x$ and $y$, assume that $A$ denotes all labels for the root-to-$y$ path,  we define an inequality $\Sigma^\mathcal{A}_r x_r \leq 1$. If this inequality cannot change the maximum solution  for the current LP problem, then the split operation in the leaf descendant axis can be safely canceled .
\label{lemma:optimization2}
\end{lemma}

 Recall \autoref{fig:new_example_mutation}.  We can safely avoid the computation for the suites (VIII), (X), (XII) and (XIV) based on \autoref{lemma:optimization2}. Combining the rule in Lemma 1 together, we  compute only 3 suites out of a total of 8, significantly reducing the search space.

One might wonder why the second optimization  would not be more aggressive to apply on all descendant axis. This is because this aggressive strategy cannot be combined with Optimization 1.  For example, consider a path pattern ``$a//b//c$". If we perform a conversion operation between $a$ and $b$, then based on Optimization 1, we cannot do the split between $b$ and $c$, leading to a suboptimal solution. Therefore,  in Optimization 2, we  consider a conservative strategy for only leaf descendant axis.

%% file: example/instance_tree.tex
\begin{subfigure}[t]{0.25\textwidth}
        \centering
        
\begin{tikzpicture}[font=\small,
	level distance=0.6cm,
  level 1/.style={sibling distance=0.5cm},
  level 2/.style={sibling distance=0.5cm},
]
\useasboundingbox [fill=gray!0] (0,0) rectangle (30mm,15mm);

\node[black] at (5mm,6mm)  {$R_0$(a,b,c)};

\draw  [<-, thick] (14mm,6mm) -- (17mm,6mm);

\node[black] at (23mm,11mm)  {a}
    child {node[black] {b}[style = {double,black}]}
    child {node[black] {c}[style = {double,black}]};

\end{tikzpicture}
        \caption{Descendant axes query}
        \label{fig:twig_with_ad}
\end{subfigure}%
\begin{subfigure}[t]{0.25\textwidth}
        \centering
        
\begin{tikzpicture}[font=\small,
	level distance=0.6cm,
  level 1/.style={sibling distance=0.5cm},
  level 2/.style={sibling distance=0.5cm},
]
\useasboundingbox [fill=gray!0] (0,0) rectangle (30mm,15mm);

\node[black] at (15mm,13mm)  {$a_1$}[style = {dotted}]
    child {node[black] {$a_n$}[style = {solid}]
    child {node[black] {$b_1$}}
    child {node[black] {...}}
    child {node[black] {$b_n$}}
    child {node[black] {$c_1$}}
    child {node[black] {...}}
    child {node[black] {$c_n$}}};

\end{tikzpicture}
        \caption{Tree instance}
        \label{fig:twig_with_ad_instance}
\end{subfigure}%

\begin{subfigure}[t]{0.25\textwidth}
        \centering
        
\begin{tikzpicture}[font=\small,
	level distance=0.6cm,
  level 1/.style={sibling distance=0.5cm},
  level 2/.style={sibling distance=0.5cm},
]
\useasboundingbox [fill=gray!0] (0,0) rectangle (30mm,15mm);

\node[black] at (5mm,6mm)  {$R_0$(a,b,c,d)};

\draw  [<-, thick] (14mm,6mm) -- (17mm,6mm);

\node[black] at (23mm,13mm)  {a}
    child {node[black] {b}
    child {node[black] {c}}}
    child {node[black] {d}};

\end{tikzpicture}
        \caption{Child axes query}
        \label{fig:twig_with_pc}
\end{subfigure}%
\begin{subfigure}[t]{0.25\textwidth}
        \centering
        
\begin{tikzpicture}[font=\small,
	level distance=0.6cm,
  level 1/.style={sibling distance=0.5cm},
  level 2/.style={sibling distance=0.5cm},
]
\useasboundingbox [fill=gray!0] (0,0) rectangle (30mm,15mm);

\node[black] at (15mm,13mm)  {$a_1$}
    child {node[black] {$b_1$}
        child {node[black] {$c_1$}}}
    child {node[black] {...}}
    child {node[black] {$b_n$}
        child {node[black] {$c_n$}}}
    child {node[black] {$d_1$}}
    child {node[black] {...}}
    child {node[black] {$d_n$}};

\end{tikzpicture}
        \caption{Tree instance}
        \label{fig:twig_with_pc_instance}
\end{subfigure}%

%% file: example/conflict_twig_instance.tex
\begin{subfigure}[t]{0.30\textwidth}
        \centering
        
\begin{tikzpicture}[font=\small,
	level distance=0.6cm,
  level 1/.style={sibling distance=0.35cm},
  level 2/.style={sibling distance=0.35cm},
]
\useasboundingbox [fill=gray!0] (0,0) rectangle (22mm,15mm);

\node[black] at (4mm,11mm)  {$R_0$(a,b,c,d)};

\draw  [<-, thick] (13mm,11mm) -- (16mm,11mm);

\node[black] at (21mm,16mm)  {a}
    child {node {b}}
    child {node {c}
        child {node[black] {d}[style = {double,black}]}};

\end{tikzpicture}
        \caption{Tree query}
        \label{fig:twig_with_confilct1}
\end{subfigure}%
\begin{subfigure}[t]{0.34\textwidth}
        \centering
        
\begin{tikzpicture}[font=\small,
	level distance=0.6cm,
  level 1/.style={sibling distance=0.35cm},
  level 2/.style={sibling distance=0.35cm},
]
\useasboundingbox [fill=gray!0] (0,0) rectangle (22mm,15mm); 

\node[black] at (11mm,16mm)  {$a_1$}
    child {node {$b_1$}}
    child {node {...}}
    child {node {$b_n$}}
    child {node {$c_1$}
        child {node[black] {$d_1$}}}
    child {node {...}}
    child {node {$c_n$}
        child {node[black] {$d_n$}}}
    ;

\end{tikzpicture}
        \caption{Tree instance 1}
        \label{fig:twig_with_confilct_instance1}
\end{subfigure}%
\begin{subfigure}[t]{0.34\textwidth}
        \centering
        
\begin{tikzpicture}[font=\small,
	level distance=0.5cm,
  level 1/.style={sibling distance=0.4cm},
  level 2/.style={level distance=0.4cm, sibling distance=0.4cm},
  level 3/.style={level distance=0.4cm, sibling distance=0.4cm},
  level 4/.style={level distance=0.5cm, sibling distance=0.4cm},
]
\useasboundingbox [fill=gray!0] (0,0) rectangle (22mm,24mm); 

\node[black] at (8mm,23mm)  {$a_1$}
    child {node {$b_1$}}
    child {node {$c_1$}
        child {node {$...$}
            child {node {$a_n$}
            child {node {$b_n$}}
            child {node {$c_n$}
                child {node {$d_1$}}
                child {node {...}}
                child {node {$d_n$}}
                }}
        }
    }
    ;

\end{tikzpicture}
        \caption{Tree instance 2}
        \label{fig:twig_with_confilct_instance2}
\end{subfigure}%

%% file: algorithm/Conversion.tex
\KwIn{A cross-model conjunctive query $Q$=\{$T, R$\}}
\KwOut{Worst-case output size bound}

$\mathbb{C} \gets CP(T, \emptyset)$ \tcp*{compute \textit{canonical suites}}
\ForEach{ $ (CT,CI) \in \mathbb{C}$}{ Generate the LP inequalities  induced by $CT, CI$ and $R$ \\ Solve the associated LP problem }

Return the maximum value for all solutions of LP problems \\

%% file: algorithm/MinorGeneration.tex
\KwIn{A tree pattern query $T$ and the associated compensation inequalities $CI$ and relations R }
\KwOut{A set of canonical suites $\mathbb{C}$}

\If{there exists any descendant axis in  $T$}{ Let $\alpha$ be a highest descendant axis  in $T$ \\
$T_1 \gets T \mid \alpha$ \tcp*{tree conversion}
 $\mathbb{C} \gets CS(T_1,CI,R)$ \\
 $T_2, T_3, CI \gets T \parallel \alpha$ \tcp*{tree split}
 
 \ForEach{ $(CT_2,CI_2,R)$ in $CS(T_2,CI,R)$}{\ForEach{ $(CT_3,CI_3,R)$ in $CS(T_3,CI,R)$}  {$\mathbb{M}$ = $\mathbb{C}$ $\cup$ \{$(CT_2 \cup CT_3, CI_2 \cup CI_3,R)$\} }  }
 return $\mathbb{C}$ \\
}
\Else{return $\big\{ (\{T\}, CI,R) \big\}$ \\}

%% file: section/approach.tex
\section{Approach}\label{sec:approach}

In this section, we tackle the challenges in designing a worst-case optimal algorithm for CMCQs over relational and tree data. We briefly review the existing relational worst-case optimal join algorithms. We represent these results informally and refer the readers to Ngo et al. \cite{DBLP:journals/sigmod/NgoRR13} for a complete survey. The first algorithm to have a running time matching these worst-case size bounds is the NPRR algorithm \cite{DBLP:journals/corr/abs-1203-1952}. An important property in NPRR is to estimate the intermediate join size and avoid to produce the case which is larger than the worst-case bound. In fact, for any join query, its execution time can be upper bounded by the AGM  \cite{DBLP:conf/focs/AtseriasGM08}. Interestingly, \textit{LeapFrog} \cite{veldhuizen2012leapfrog} and \textit{Joen}  \cite{ciucanu2015worst} completely abandon the ``query plan'' and propose to deal with one attribute at a time with multiple relations at the same time. 

\begin{figure}[t!]
     \centering
    \small
     \input{example/algorithm_joining_order}
     \caption{A CMCQ (a) and its table instance (b) and tree instance (c).}
     \label{fig:alg_cross_model_query_example}
\end{figure}
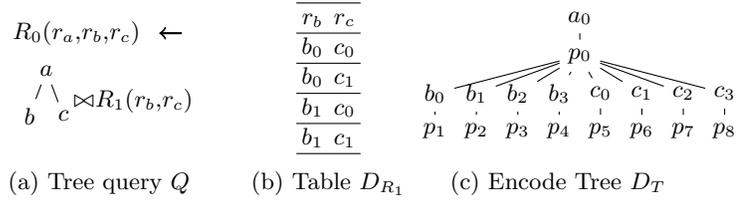

\subsection{Tree and relational data representation} 

To answer a tree pattern query,  a positional representation of occurrences of tree elements and string values in the tree database are widely used, which extends the classic inverted index data structure in information retrieval. There existed two common ways to encode an instance tree, i.e. Dewey encoding~\cite{DBLP:conf/vldb/LuLCC05} and containment encoding \cite{DBLP:conf/sigmod/BrunoKS02}. These decodings are necessary as they allow us to partially join tree patterns to avoid undesired intermediate result. After encoding, each attribute $j$ in the query node can be represented as a \textbf{node table} in form of $t_j(r_j, p_j)$, where $r_j$ and $p_j$ are the label value and position value, respectively. Check an example from a encoded tree instance in \autoref{fig:alg_jo_tree_instance}. The position value can be added in $\mathcal{O}(N)$ by one scan of the original tree. Note that we use Dewey coding in our implementation but the following algorithm is not limited to such representation. Any representation scheme which captures the structure of trees such as a region encoding \cite{DBLP:conf/sigmod/BrunoKS02} and an extended Dewey encoding \cite{DBLP:conf/vldb/LuLCC05} can all be applied in the algorithm.

All the data in relational are label data, and all relation tables and node tables will be expressed by the Trie index structure, which is commonly applied in the relational worst-case optimal algorithms (e.g. \cite{DBLP:conf/sigmod/AbergerTOR16,veldhuizen2012leapfrog}).
The Trie structure can be accomplished using standard data structures (notably, balanced trees with $\mathcal{O}(\log n)$ look-up time or nested hashed tables with $\mathcal{O}(1)$ look-up time).

\subsection{Challenges}

In our context, tree data and twig pattern matching do make situation more complex. Firstly, directly materializing tree pattern matching may yield asymptotically more intermediate results. If we ignore the pattern, we may loose some bound constraints. Secondly, since tree data are representing both label and position values, position value joining may require more computation cost for pattern matching while we do not need position values in our final result.

\begin{example}\label{example:worst-case-motivation}
Recall that a triangle relational join query $Q$=$R_1(r_a, r_b) \bowtie R_2(r_a, r_c) $ $\bowtie R_3(r_b, r_c)$ has size bound $\mathcal{O}(N^{\frac{3}{2}})$. \autoref{fig:alg_cross_model_query_example} depicts an example of a CMCQ $Q$ with the table $R_1(r_b, r_c)$ and twig query $a[b]/c$ to return result $R_0(r_a, r_b, r_c)$, which also has size bound $\mathcal{O}(N^{\frac{3}{2}})$ since the PC paths $a/b$ and $a/c$ are equivalent to the constraints $x_a + x_b \leq 1$ and $x_a + x_c \leq 1$, respectively. \autoref{fig:alg_jo_table_instance} and \autoref{fig:alg_jo_tree_instance} show the instance table $D_{R_1}$ and the encoded tree $D_T$. The number of label values in the result $R_0(r_a,r_b,r_c)$ is only $4$ rows which is $\mathcal{O}(N)$. On the other hand, the result size of only the tree pattern is $16$ rows which is $\mathcal{O}(N^2)$, where $N$ is a table size or a node size for each attribute. The final result with the position values is also $\mathcal{O}(N^2)$. Here, $\mathcal{O}(N^2)$ is from the matching result of the position values of the attributes $t_b$ and $t_c$.
\end{example}

EmptyHeaded \cite{DBLP:conf/sigmod/AbergerTOR16} applied the existing worst-case optimal algorithms to process the graph edge pattern matching. We may also attempt to solve relation-tree joins by representing the trees as relations with the node-position and the node-label tables and then reformulating the cross-model conjunctive query as a relational conjunctive query. However, as \autoref{example:worst-case-motivation} illustrated, such method can not guarantee the worst-case optimality as extra computation is required for position value matching in a tree. 

\subsection{Cross-model join (CMJoin) algorithm} 

In this part, we discuss the algorithm to process both relational and tree data. As the position values are excluded in the result set while being required for the tree pattern matching, our algorithm carefully deals with it during the join. We propose an efficient cross-model join algorithm called \textit{CMJoin} (cross-model join). In certain cases it guarantees the runtime optimality.
We discover the join result size under three scenarios: with all node position values, with only branch node position value, and without position value.

\begin{lemma}
Given relational tables $\mathcal{R}$ and pattern queries $\mathcal{T}$, let $S_r$, $S_p$, and $S_p'$ be the sets of all relation attributes, all position attributes, and only branch node position attributes, respectively. Then it holds that
\begin{equation}
    \rho_1(S_r\cup S_p) \geq \rho_2(S_r\cup S_p') \geq \rho_3(S_r).
\end{equation}
\end{lemma}
\begin{proof}
$Q(S_r)$ is the projection result from $Q(S_r\cup S_p')$ by removing all position values, and $Q(S_r\cup S_p')$ is the projection result from $Q(S_r\cup S_p)$ by removing non-branch position values. Therefore, the result size holds $\rho_1(S_r\cup S_p) \geq \rho_2(S_r\cup S_p') \geq \rho_3(S_r)$.
\end{proof}

\begin{example}
Recall the CMCQ $Q$ in \autoref{fig:alg_jo_query}, which is $Q$=$R_1(r_b, r_c) \bowtie a[b]/c$. Nodes $a$, $b$, $c$ in the tree pattern can be represented as node tables $(r_a,p_a)$, $(r_b,p_b)$, and $(r_c,p_c)$, respectively. So we have $Q(S_r\cup S_p) = R(r_a,r_b,r_c,p_a,p_b,p_c)$, $Q(S_r\cup S_p') = R(r_a,r_b,r_c,p_a)$, and $Q(S_r) = R(r_a,r_b,r_c)$. By the LP constraint bound for the relations and PC-paths, we achieve $\mathcal{O}(N^2)$, $\mathcal{O}(N^\frac{3}{2})$, and $\mathcal{O}(N^\frac{3}{2})$ for the size bounds $\rho_1(S_r\cup S_p)$, $\rho_2(S_r\cup S_p')$, and $\rho_3(S_r)$, respectively. 
\end{example}

We elaborate \textit{CMJoin} \autoref{alg:cmjoin} more in the following. In the case of $\rho_1(S_r\cup S_p)$=$\rho_3(S_r)$, \textit{CMJoin} executes a generic relational worst-case optimal join algorithm~\cite{DBLP:conf/sigmod/AbergerTOR16,ciucanu2015worst} as the extra position values do not affect the worst-case final result. In other cases, \textit{CMJoin} computes the path result of the tree pattern first. In this case, we project out all position values of a non-branch node for the query tree pattern. Then, we keep the position values of the only branch node so that we still can match the whole part of the tree pattern. 

\begin{algorithm} \caption{\textit{CMJoin}}
\label{alg:cmjoin}
\input{algorithm/cmjoin.tex}
\end{algorithm}

\begin{theorem}\label{theo:optimal_algorithm1}
Assume we have relations $\mathcal{R}$ and pattern queries $\mathcal{T}$. If either
\begin{enumerate}
    \item[(1)] $\rho_1(S_r\cup S_p) \leq \rho_{3}(S_r)$ or
    \item[(2)] (i) $\rho_1(S_r\cup S_p') \leq \rho_{3}(S_r)$ and (ii) for each path $P$ in $\mathcal{T}$ let $S_r''$ and $S_p''$ be the set of label and position attributes for $P$ so that $\rho_4(S_r''\cup S_p'') \leq \rho_{3}(S_r)$.
\end{enumerate}
Then, \textit{CMJoin} is worst-case optimal to $\rho_{3}(S_r)$.
\end{theorem}

\begin{proof}
\textit{(1)} Since the join result of the only label value $\rho_{3}(S_r)$  is the projection of $\rho_1(S_r\cup S_p)$, we can compute $Q(S_r\cup S_p)$ first, then project out all the position value in linear of $\mathcal{O}(N^{\rho_1(S_r\cup S_p)})$. Since $\rho_1(S_r\cup S_p) \leq \rho_{3}(S_r)$, we can estimate that the result size is limited by $\mathcal{O}(N^{\rho_{3}(S_r)})$.

\textit{(2)} $\rho_1(S_r''\cup S_p'') \leq \rho_{3}(S_r)$ means that each path result with label and position values are under worst-case result of $\rho_{3}(S_r)$. We may first compute the path result and then project out all the non-branch position values. The inequality $\rho_2(r,p') \leq \rho_3(r)$ means that the join result containing all branch position values has a worst-case result size which is still under $\rho_{3}(S_r)$. Then by considering those position values as relational attribute values and by a generic relation join \cite{DBLP:journals/corr/abs-1203-1952,DBLP:conf/sigmod/AbergerTOR16}, 
\textit{CMJoin} is worst-case optimal to $\rho_{3}(S_r)$.
\end{proof}

\begin{example}
Recall the CMCQ query $Q = R_1(r_b, r_c) \bowtie a[b]/c$ in \autoref{fig:alg_jo_query}. Since $\rho_1(S_r\cup S_p)>\rho_3(r)$, directly computing all label and position values may generate asymptotically bigger result ($\mathcal{O}(N^2)$ in this case). So we can compute path results of $a/b$ and $a/c$, which are $(r_a,p_a,r_b,p_b)$ and $(r_a,p_a,r_c,p_c)$ and in $\mathcal{O}(N)$. Then we obtain only branch node results $(r_a,p_a,r_b)$ and $(r_a,p_a,r_c)$. By joining these project-out result with relation $R_1$ by a generic worst-case optimal algorithm, we can guarantee the size bound is $\mathcal{O}(N^\frac{3}{2})$.
\end{example}

%% file: example/algorithm_joining_order.tex
\begin{subfigure}[t]{0.28\textwidth}
        \centering
        
\begin{tikzpicture}[font=\small,
	level distance=0.6cm,
  level 1/.style={sibling distance=0.45cm},
  level 2/.style={sibling distance=0.45cm},
]
\useasboundingbox [fill=gray!0] (0,0) rectangle (22mm,20mm);

\node[black] at (8mm,16mm)  {$R_0$($r_a$,$r_b$,$r_c$)};

\draw  [<-, thick] (19mm,16mm) -- (22mm,16mm);

\node[black] at (4mm,11mm)  {$a$}
    child {node {$b$}}
    child {node {$c$}};

\node[black] at (9mm,7mm)  {$\bowtie$};
\node[black] at (17mm,7mm)  {$R_1$($r_b$,$r_c$)};
\end{tikzpicture}
        \caption{Tree query $Q$}
        \label{fig:alg_jo_query}
\end{subfigure}%
\begin{subfigure}[t]{0.22\textwidth}
        \centering
        
\begin{tikzpicture}[font=\small,
	level distance=0.6cm,
  level 1/.style={sibling distance=0.35cm},
  level 2/.style={sibling distance=0.35cm},
]
\useasboundingbox [fill=gray!0] (0,0) rectangle (10mm,20mm); 

\node[black] at (5mm,10mm) {
\setlength{\tabcolsep}{2pt}
\begin{tabular}{cc}
\hline
\textbf{$r_b$} & \textbf{$r_c$}\\ \hline
$b_0$     & $c_0$           \\ \hline
$b_0$       & $c_1$         \\ \hline  
$b_1$      & $c_0$         \\ \hline
$b_1$       & $c_1$        \\ \hline 
\end{tabular}
\medskip
};

\end{tikzpicture}
        \caption{Table $D_{R_1}$}
        \label{fig:alg_jo_table_instance}
\end{subfigure}%
\begin{subfigure}[t]{0.28\textwidth}
        \centering
        
\begin{tikzpicture}[font=\small,
	level distance=0.5cm,
  level 1/.style={sibling distance=2cm},
  level 2/.style={level distance=0.5cm, sibling distance=0.55cm}
]
\useasboundingbox [fill=gray!0] (0,0) rectangle (40mm,20mm); 

\node[black] at (20mm,18mm)  {$a_0$}
    child {node {$p_0$}
        child {node {$b_0$} child {node {$p_1$}}}
        child {node {$b_1$} child {node {$p_2$}}}
        child {node {$b_2$} child {node {$p_3$}}}
        child {node {$b_3$} child {node {$p_4$}}}
        child {node {$c_0$} child {node {$p_5$}}}
        child {node {$c_1$} child {node {$p_6$}}}
        child {node {$c_2$} child {node {$p_7$}}}
        child {node {$c_3$} child {node {$p_8$}}}}
    ;

\end{tikzpicture}
        \caption{Encode Tree $D_T$}
        \label{fig:alg_jo_tree_instance}
\end{subfigure}%

%% file: algorithm/cmjoin.tex
    \KwIn{Relational tables $\mathcal{R}$, pattern queries $\mathcal{T}$}
    \DontPrintSemicolon
    $\mathcal{R}'\gets \emptyset$\tcp*{Tree intermediate result}
    
    \eIf(\tcp*[h]{\autoref{theo:optimal_algorithm1} condition (1)}){$\rho_{1}(S_r\cup S_p) \leq \rho_{3}(S_r)$}{
       \ForEach{ $ N \in \mathcal{T}$}{ 
        $\mathcal{R}' \gets \mathcal{R}' \cup C_N(r_N, p_N)$  \tcp*{Nodes as tables}
        }
       }{
    $\mathcal{P} \gets \mathcal{T}.getPaths()$\\
    \ForEach{ $ P \in \mathcal{P}$}{ 
    $R_P(S_r\cup S_p) \gets $ path result of $P$  \tcp*{Paths as tables}
    $R_P'(S_r\cup S_p') \gets$ project out non-branch position values of $R_P(S_r\cup S_p)$\\
    $\mathcal{R}' \gets \mathcal{R}' \cup \{R_P'(S_r\cup S_p')\}$
    }
    }

    $Q(S_r\cup S_p') \gets generic\_join(\mathcal{R} \cup \mathcal{R}')$
    
    $Q(S_r) \gets$ project out all position values $Q(S_r\cup S_p')$\\
    \KwOut{Join results $Q(S_r)$}
    

%% file: section/evaluations.tex
\begin{table}[ht!]
    \footnotesize
    \scriptsize
    \centering
    \caption{Intermediate result size ($10^6$) and running time (S) for queries. ``/'' and ``-'' indicate ``timeout'' ($\geq$ 10 mins) and ``out of memory''. We measure the intermediate size by accumulating all intermediate and final join results.}
    \label{tab:intermediate_result}

\input{table/intermediate}
\end{table}

\section{Evaluation}\label{sec:evaluation}
In this section, we experimentally evaluate the performance of the proposed algorithms and \textit{CMJoin} with four real-life and benchmark data sets. We comprehensively evaluate \textit{CMJoin} against state-of-the-art systems and algorithms concerning efficiency, scalability, and intermediate cost.

\subsection{Evaluation setup}
\noindent \textbf{Datasets and query design} \indent \autoref{tab:dataset} provides the statistics of datasets and designed CMCQs. These diverse datasets differ from each other in terms of the tree structure, data skewness, data size, and data model varieties. Accordingly, we designed 24 CMCQs to evaluate the efficiency, scalability, and cost performance of the \textit{CMJoin} in various real-world scenarios. 

\noindent\textbf{Comparison systems and algorithms} \indent 
\textit{CMJoin} is compared with two types of state-of-the-art cross-model solutions. The first solution is to use one query to retrieve a result without changing the nature of models \cite{DBLP:conf/fnc/NassiriMH18,zhangunibench}. We implemented queries in PostgreSQL (\textit{PG}), that supports cross-model joins. This enables the usage of the \textit{PG}'s default query optimizer. 

The second solution is to encode and retrieve tree nodes in a relational engine \cite{DBLP:journals/jsw/BousalemC15,DBLP:journals/kbs/QtaishA16,DBLP:conf/ACISicis/ZhuYFS17,DBLP:conf/sigmod/AbergerTOR16}. We implemented two algorithms, i.e. structure join (\textit{SJ}) (pattern matching first, then matching the between values) and value join (\textit{VJ}) (label value matching first, then matching the position values). Also, we compared to a worst-case optimal relational engine called EmptyHeaded (\textit{EH}) \cite{DBLP:conf/sigmod/AbergerTOR16}.

\noindent \textbf{Experiment Setting} \indent
We conducted all experiments on a 64-bit Windows machine with a 4-core Intel i7-4790 CPU at 3.6GHz, 16GB RAM, and 500GB HDD. We implemented all solutions, including \textit{CMJoin} and the compared algorithms, in memory processing by Python 3. We measured the computation time of joining as the main metric excluding the time used for compilation, data loading, index presorting, and representation/index creation for all the systems and algorithms. We employed the Dewey encoding \cite{DBLP:conf/vldb/LuLCC05} in all experiments. The join order of attributes is greedily chosen based on the frequency of attributes. We measured the intermediate cost metric by accumulating all intermediate and final join results. For \textit{PG} we accumulated all sub-query intermediate results. 
We repeated five experiments excluding the lowest and the highest measure and calculated the average of the results. Between each measurement of queries we wiped caches and re-loaded the data to avoid intermediate results.

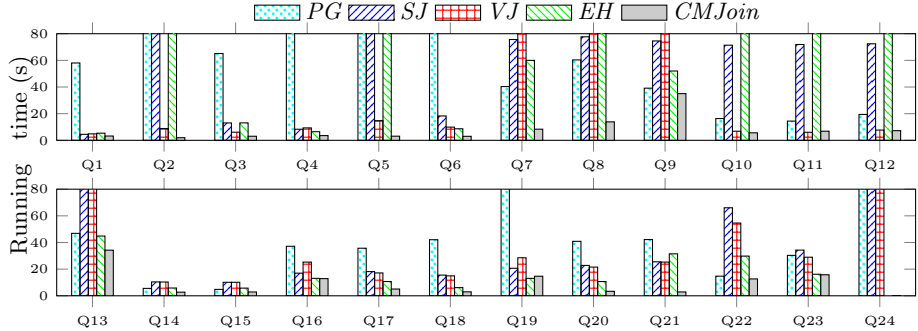
\begin{figure}[t]
    \centering
    \input{evaluation/efficiency}
    \caption{\textbf{Efficiency:} runtime performance for all queries by $PG$, $SJ$, $VJ$, $EH$, and $CMJoin$. The performance time of more than 80s is cut for better presentation.} 
    \label{fig:evaluation_efficiency}
\end{figure}


\noindent \textbf{Efficiency} \indent \autoref{fig:evaluation_efficiency} shows the evaluation of the efficiency. In general, \textit{CMJoin} is $3.33$-$13.43$ times faster in average than other solutions as shown in \autoref{tab:intermediate_result}. These numbers are conservative as we exclude the ``out of memory'' (OOM) and ``time out'' (TO) results from the average calculation. Algorithms \textit{SJ}, \textit{VJ}, and \textit{EH} perform relatively better compared to \textit{PG} in the majority of the cases as they encoded the tree data into relation-like formats, making it faster to retrieve the tree nodes and match twig patterns. 

Specifically in queries $Q1$-$Q6$, \textit{CMJoin}, \textit{SJ}, \textit{VJ}, and \textit{EH} perform better than \textit{PG}, as the original tree is deeply recursive in the TreeBank dataset \cite{xue2005penn_treebank}, and designed tree pattern queries are complex. So, it is costly to retrieve results directly from the original tree by \textit{PG}. Instead, \textit{CMJoin}, \textit{SJ}, and \textit{VJ} use encoded structural information to excel in retrieving nodes and matching tree patterns in such cases. 
In $Q2$ and $Q5$, \textit{EH} performs worse. The reason is that it seeks for a better instance bound by joining partial tables and sub-twigs first and then aggregates the result. However, the separated joins yield more intermediate result in such cases in this dataset.
In $Q1$ and $Q4$, which deal with a single table, \textit{SJ} and \textit{VJ} perform relatively close as no table joining occurs in these cases. However, in $Q2$-$Q3$ and $Q5$-$Q6$, \textit{SJ} performs worse as joining two tables first leads to huge intermediate results in this dataset. 

In contrast to the above, \textit{SJ} outperforms \textit{VJ} in $Q7$-$Q9$. The reason is that in the Xmark dataset \cite{DBLP:conf/vldb/SchmidtWKCMB02}, the tree data are flat and with less matching results in twig queries. The data in tables are also less skew. Therefore, \textit{SJ} operates table joins and twig matching separately yielding relatively low results. Instead, \textit{VJ} considers tree pattern matching later yielding too many intermediate results (see details of $Q7$ in \autoref{fig:evaluation_scalability} and \autoref{fig:evaluation_intermediate}) when joining label values between two models with non-uniform data. \textit{PG}, which implements queries in a similar way of \textit{SJ}, performs satisfactorily as well. The above comparisons show that compared solutions, which can achieve superiority only in some cases, and can not adapt well to dataset dynamics. 

Queries $Q10$-$Q12$ have more complex tree pattern nodes involved. In these cases \textit{VJ} filters more values and produces less intermediate results. Thus it outperforms \textit{SJ} ($\sim$10x) and \textit{PG} ($\sim$2x). For queries $Q10$-$Q12$ \textit{EH} also yields huge intermediate results with more connections in attributes. The comparison between $Q7$-$Q9$ and $Q10$-$Q12$ indicates that the solutions can not adapt well to query dynamics.

Considering queries $Q13$-$Q15$ and $Q22$-$Q24$, \textit{PG} performs relatively well since it involves only JSON and relational data. \textit{PG} performs well in JSON retrieving because JSON documents have a simple structure. In $Q14$-$Q15$ most of the solutions perform reasonably well when the result size is small but \textit{SJ}, and \textit{VJ} still suffer from a large result size in $Q13$. With only JSON data, \textit{SJ} and \textit{VJ} perform similarly, as they both treat a simple JSON tree as one relation. In contrast in $Q16$-$Q21$, it involves XML, JSON and relational data from the UniBench dataset \cite{zhangunibench}. \textit{CMJoin}, \textit{SJ}, \textit{VJ}, and \textit{EH} perform better than \textit{PG}. This is again because employing the encoding technique in trees accelerates node retrieval and matching tree patterns. Also, \textit{CMJoin}, \textit{SJ}, \textit{VJ}, and \textit{EH} are able to treat all the data models together instead of achieving results separately from each model by queries in \textit{PG}. 

Though compared systems and algorithms possess their advantages of processing and matching data, they straightforwardly join without bounding intermediate results, thus achieving sub-optimal performance during joining. \textit{CMJoin} is the clear winner against other solutions, as it can wisely join between models and between data to avoid unnecessary quadratic intermediate results. 

\begin{figure}[ht]
    \centering
    \input{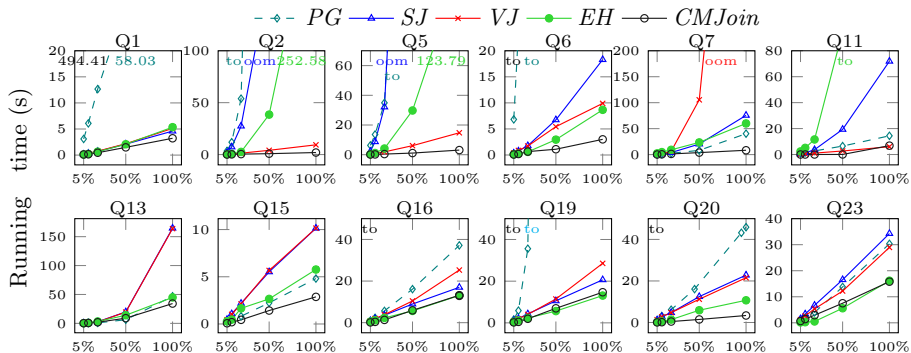}
    \caption{\textbf{Scalability:} runtime performance of \textit{PG}, \textit{SJ}, \textit{VJ}, \textit{EH}, and \textit{CMJoin}. The x-axis is the percentage of data size. 
    ``oom'' and ``to'' stand for ``out of memory'' error and ``timeout'' ($\geq$ 10 mins), respectively.}
    \label{fig:evaluation_scalability}
\end{figure}

\noindent \textbf{Scalability} \indent \autoref{fig:evaluation_scalability} shows the scalability evaluation. In most queries, \textit{CMJoin} performs flatter scaling as data size increases because \textit{CMJoin} is designed to control the unnecessary intermediate output. 

As discussed, \textit{CMJoin}, \textit{SJ}, \textit{VJ}, and \textit{EH} outperform \textit{PG} in most of the queries, as the encoding method of the algorithms speeds up the twig pattern matching especially when the documents or queries are complex. However, \textit{PG} scales better when involving simpler documents (e.g. in $Q15$ and $Q23$) or simpler queries (e.g. in $Q7$). Comparing to processing XML tree pattern queries, \textit{PG} processes JSON data more efficiently.

Interestingly in $Q2$, \textit{SJ} and \textit{PG} join two relational tables separately from twig matching, generating quadratic intermediate results, thus leading to the OOM and TO, respectively. In $Q7$ \textit{VJ} joins tables with node values without considering tree pattern structural matching and outputs an unwanted non-linear increase of intermediate results, thus leading to OOM in larger data size. Likewise evaluating \textit{EH} between $Q2$ and $Q7$, it can not adapt well with different datasets.
Performing differently in diverse datasets between \textit{SJ}/\textit{PG} and \textit{VJ}/\textit{EH} indicates that they can not smartly adapt to dataset dynamics. While increasing twig queries in $Q11$ compared to $Q7$, \textit{VJ} filters more results and thus decreases the join cost and time in $Q11$. The comparison between \textit{SJ}/\textit{EH} and \textit{VJ} shows dramatically different performance in the same dataset with different queries that indicates they can not smartly adapt to query dynamics.

In $Q11$, both \textit{CMJoin} and \textit{VJ} perform efficiently as they can filter out most of the values at the beginning. 
In this case, \textit{CMJoin} runs slightly slower than \textit{VJ}, which is reasonable as \textit{CMJoin} maintains a tree structure whereas \textit{VJ} keeps only tuple results. Overall, \textit{CMJoin} judiciously joins between models and controls unwanted massive intermediate results. The evaluation shows that it performs efficiently and stably in dynamical datasets, with various queries and it also scales well.

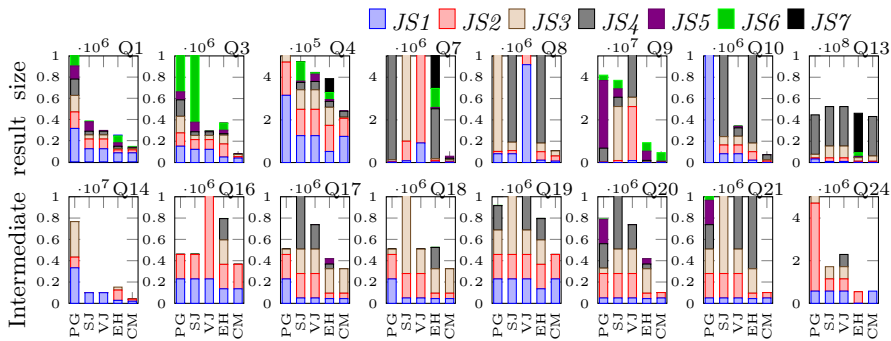
\begin{figure}[t]
    \centering
    \input{evaluation/intermediate}
    \caption{\textbf{Cost:} intermediate result size. CM: \textit{CMJoin}. JS: Joining step.}
    \label{fig:evaluation_intermediate}
\end{figure}

\noindent \textbf{Cost analysis} \indent
\autoref{tab:intermediate_result} presents the intermediate result sizes showing that \textit{CMJoin} outputs 5.46x, 5.90x, 1.92x, and 1.90x less intermediate results on average than \textit{PG}, \textit{SJ}, \textit{VJ}, and \textit{EH}, respectively. \autoref{fig:evaluation_intermediate} depicts more detailed intermediate results for each joining step. In general, \textit{CMJoin} generates less intermediate results due to its designed algorithmic process, worst-case optimality, as well as join order selections. In contrast, \textit{PG}, \textit{SJ}, and \textit{VJ} can easily yield too many (often quadratic) intermediate results during joining in different datasets or queries. This is because they have no technique to avoid undesired massive intermediate results.

\textit{PG} and \textit{SJ} suffer when the twig matching becomes complex in datasets (e.g. $Q3$ and $Q10$), while \textit{VJ} suffers in the opposite case of simpler twig pattern matching (e.g. $Q7$ and $Q16$). More specifically in $Q3$, \textit{PG} and \textit{SJ} output significant intermediate results by joining of two relational tables. In turn, \textit{VJ} controls intermediate results utilizing the values of common attributes and tags between two models. On the other hand, in $Q7$, $Q9$, and $Q16$ \textit{VJ} does not consider structural matching at first yielding unnecessary quadratic intermediate results. The above two-side examples indicate that solutions considering only one model at a time or joining values first without twig matching produce an undesired significant intermediate result. 

\textit{EH} suffers when the queries and attributes are more connected that leads to larger intermediate results during join procedures. The reason is that \textit{EH} seeks a better instance bound so that it follows the query plan based on the GHD decomposition \cite{DBLP:conf/sigmod/AbergerTOR16}. Our proposed method, \textit{CMJoin}, by wisely joining between models, avoids an unnecessary massive intermediate output from un-joined attributes. 

\noindent \textbf{Summary} \indent We summarize evaluations of \textit{CMJoin} as follows:
\begin{enumerate}[leftmargin=*]
\item Extensive experiments on diverse datasets and queries show that averagely \textit{CMJoin} achieves up to 13.43x faster runtime performance and produces up to 5.46x less intermediate results compared to other solutions.
\item With skew data \textit{CMJoin} avoids undesired huge intermediate results by wisely joining data between models. With uniform data \textit{CMJoin} filters out more values by joining one attribute at a time between all models.
\item With more tables, twigs, or common attributes involved \textit{CMJoin} seems to perform more efficiently and scale better.
\end{enumerate}

%% file: table/intermediate.tex
\setlength{\tabcolsep}{1pt}
\begin{tabular}{c|ccccc|ccccc}

\toprule
 & \multicolumn{5}{c|}{Intermediate result size($10^6$)} & \multicolumn{5}{c}{Running time (second)} \\ \hline
Query & \textbf{PG} & \textbf{SJ} & \textbf{VJ} & \textbf{EH} & \textbf{\textit{CMJoin}}  & \textbf{PG} & \textbf{SJ} & \textbf{VJ} & \textbf{EH} & \textbf{\textit{CMJoin}} \\ \hline
Q1 & 7.87x & 2.60x & 2.00x & 1.68x & 0.15 & 18.02x & 1.39x & 1.51x & 1.66x & 3.22 \\
Q2 & / & - & 3.75x & 4.83x & 0.08 & / & - & 4.52x & 129x & 1.96 \\
Q3 & 86.0x & 62.6x & 3.63x & 4.61x & 0.08  & 21.3x & 4.27x & 1.99x & 4.28x & 3.06 \\
Q4 & / & 1.96x & 1.75x & 1.64x & 0.24  & / & 2.34x & 2.63x & 1.82x & 3.55 \\
Q5 & / & - & 1.86x & 1.77x & 0.22  & / & - & 4.75x & 39.8x & 3.11 \\
Q6 & / & 2.24x & 2.00x & 1.85x & 0.21  & / & 6.10x & 3.30x & 2.89x & 3.00 \\ \hline
Q7 & 133x & 106x & - & 35.1x & 0.29 & 4.82x & 9.05x & - & 7.18x & 8.36 \\
Q8 & 350x & 279.8x & - & / & 0.11 & 4.36x & 5.61x & - & / & 13.8 \\
Q9 & 8.87x & 8.34x & - & 2.01x & 4.62 & 1.12x & 2.13x & - & 1.48x & 35.0 \\
Q10 & 110x & 440x & 4.86x & / & 0.07 & 2.91x & 12.7x & 1.22x & / & 5.62 \\
Q11 & 110x & 440x & 4.86x & / & 0.07 & 2.11x & 10.5x & 0.88x & / & 6.84 \\
Q12 & 110x & 440x & 4.86x & / & 0.07 & 2.68x & 9.99x & 1.06x & / & 7.25 \\ \hline
Q13 & 1.04x & 1.22x & 1.22x & 1.07x & 43.2 & 1.37x & 4.81x & 4.79x & 1.31x & 34.2 \\
Q14 & 19.7x & 2.56x & 2.56x & 3.90x & 0.39 & 2.04x & 3.82x & 3.79x & 2.14x & 2.73 \\
Q15 & 14.2x & 1.85x & 1.85x & 17.0x & 0.54 & 1.68x & 3.53x & 3.54x & 2.01x & 2.87 \\
Q16 & 1.24x & 1.24x & 6.81x & 2.15x & 0.37  & 2.88x & 1.32x & 1.96x & 1.02x & 12.9 \\
Q17 & 1.59x & 7.84x & 2.28x & 1.31x & 0.32  & 7.03x & 3.58x & 3.38x & 2.10x & 5.08 \\
Q18 & 1.59x & 7.13x & 1.59x & 1.64x & 0.32  & 14.1x & 5.21x & 5.02x & 2.06x & 2.98 \\
Q19 & / & 5.47x & 6.62x & 1.77x & 0.45 & / & 1.41x & 1.94x & 0.89x & 14.7 \\
Q20 & 7.80x & 25.1x & 7.30x & 4.19x & 0.10 & 12.1x & 6.77x & 6.39x & 3.17x & 3.37 \\
Q21 & 12.0x & 36.1x & 18.4x & 14.7x & 0.10 & 14.6x & 8.82x & 8.75x & 10.9x & 2.89 \\ \hline
Q22 & 1.00x & 18.5x & 18.5x & 0.96x & 0.57 & 1.16x & 5.22x & 4.31x & 2.35x & 12.7 \\
Q23 & 18.5x & 18.5x & 18.5x & 1.61x & 0.57 & 1.92x & 2.17x & 1.83x & 1.02x & 15.8 \\
Q24 & 14.3x & 3.02x & 4.02x & 0.96x & 0.57 & $>$9kx &$>$11kx & $>$12kx & 0.18x & 0.01 \\ \hline
AVG & 5.46x & 5.90x & 1.92x & 1.90x & 2.24 & 4.37x & 5.34x & 3.33x & 3.46x & 8.54 \\
\bottomrule
\end{tabular}

%% file: evaluation/efficiency.tex
\centering
\begin{tikzpicture}

\pgfplotstableread{
x y
1	494.4125671
2	1000
3	1000
4   1000
5   1000
6   1000
7	128.9668016
8	1000
9	77.19078016
10   1000
11   1000
12   1000
13	235.0922918
14	8.783295393
15	4.018759727
16	1000
17	1000
18	1000
19	1000
20	1000
21	1000
22  18.647146940231323
23  23.161616086959839
24  42.528913974761963
}{\cdb}

\pgfplotstableread{
x y
1	58.03966451
2	1000
3	65.11681032
4   111.46849179267883
5   1000
6   1000
7	40.32632518
8	60.33269191
9	39.12628984
10  16.333496809005737
11  14.419880867004395
12  19.42306113243103
13	46.8560257
14	5.562289
15	4.81264739
16	37.1566534
17	35.70338607
18	42.09003496
19  1000
20  40.87788724899292
21  42.18544340133667
22  14.743799686431885
23  30.32495856285095215
24  96.38151056383052
}{\pg}

\pgfplotstableread{
x y
1	4.494897366
2	1000
3	13.05675936
4   8.315558910369873
5   1000
6   18.29512310028076
7	75.62656379
8	77.68454003
9	74.57444024
10  71.37913632392883
11  71.89853930473328
12  72.4368953704834
13	164.737045
14	10.42244792
15	10.13814878
16	16.98677325
17	18.16654372
18	15.53062129
19  20.65147566795349
20  22.805636167526245
21  25.502328872680664
22  66.05639863014221
23  34.27667188644409
24  110.76280808448792
}{\nj}

\pgfplotstableread{
x y
1	4.848903418
2	8.859442949
3	6.100447655
4   9.328198432922363
5   14.762009859085083
6   9.902380228042603
7	1000
8	1000
9	1000
10  6.83014988899231
11  6.020454168319702
12  7.7103517055511475
13	163.9085462
14	10.35088587
15	10.1486249
16	25.29368258
17	17.14656782
18	14.96762347
19  28.544870853424072
20  21.53696608543396
21  25.296306610107422
22  54.58181047439575
23  28.9535551071167
24  121.76718211174011
}{\vj}

\pgfplotstableread{
x y
1	5.349941968917847
2	252.5817449092865
3	13.087437391281128
4   6.448789834976196
5   123.79373908042908
6   8.667698860168457
7	60.002097368240356
8   971.423177242279
9	52.00543999671936
10  971.423177242279
11  971.423177242279
12  971.423177242279
13	44.85323095321655
14	5.8465353012085
15	5.76629590988159
16	13.086905241012573
17	10.690909624099731
18	6.135972023010254
19  13.086905241012573
20  10.690909624099731
21  31.49009609222412
22  29.792332887649536
23  16.163808345794678
24  0.01/
}{\eh}

\pgfplotstableread{
x y
1	3.221966314
2	1.957942963
3	3.061967659
4   3.5455145835876465
5   3.107366371154785
6   3.0015110969543457
7	8.356827736
8	13.85473633
9	35.03611684
10  5.624425506591797
11  6.841535568237305
12  7.245246648788452
13	34.24016118
14	2.730003595
15	2.87126708
16	12.88750839
17	5.07830677
18	2.981938601
19  14.688944101333618
20  3.370041608810425
21  2.894787216186523
22  12.663344383239746
23  15.796029567718506
24  0.0009980201721191406
}{\xjoin}

\begin{groupplot}[
        width=13cm, 
		height=3cm, 
		label style={font=\tiny},
 		ticklabel style={font=\tiny},
        group style={group name=plots,group size=1 by 2, horizontal sep=2em, vertical sep=2em}]
 \nextgroupplot[
        ybar,
        ybar=.001cm,
        bar width=0.11cm,
        xtick={1,...,24},
        ylabel=Running~~~time (s),
        ylabel style={font=\small},
        xticklabels={Q1,Q2,Q3,Q4,Q5,Q6,Q7,Q8,Q9,Q10,Q11,Q12,Q13,Q14,Q15,Q16,Q17,Q18,Q19,Q20,Q21,Q22,Q23,Q24},
       legend columns=-1,
       legend style={at={(0.22,1)},anchor=south west,draw=none,fill=none,font=\small},
      ylabel style={at={(-0.75cm,-0.2cm)},anchor=north,font=\small},
         legend image code/.code={%
                    \draw[#1, draw=black] (0cm,-0.1cm) rectangle (0.5cm,0.2cm);
                },
        area legend,
       xmin=0.5,
       xmax=12.5,
      ymax = 80,
      ymin=0,
]

         
        
         
         \addplot[ybar, pattern color=cyan!150, pattern=crosshatch dots] table[
           x expr=\thisrowno{0}, 
           y expr=\thisrowno{1}]{\pg};\addlegendentry{\textit{PG}};
        \addplot[ybar, pattern color=blue!150, pattern=north east lines, ] table[
           x expr=\thisrowno{0}, 
           y expr=\thisrowno{1}]{\nj};\addlegendentry{\textit{SJ}};
        \addplot[ybar, pattern color=red!150, pattern=grid] table[
           x expr=\thisrowno{0}, 
           y expr=\thisrowno{1}]{\vj};\addlegendentry{\textit{VJ}};
          
        \addplot[ybar, pattern color=green!150, pattern=north west lines] table[
           x expr=\thisrowno{0}, 
           y expr=\thisrowno{1}]{\eh};\addlegendentry{\textit{EH}};
        \addplot[ybar,fill=black!20] table[
           x expr=\thisrowno{0}, 
           y expr=\thisrowno{1}]{\xjoin};\addlegendentry{\textit{CMJoin}};

 \nextgroupplot[
        ybar,
        ybar=.001cm,
        bar width=0.11cm,
        xtick={1,...,24},
        ylabel style={font=\small},
        xticklabels={Q1,Q2,Q3,Q4,Q5,Q6,Q7,Q8,Q9,Q10,Q11,Q12,Q13,Q14,Q15,Q16,Q17,Q18,Q19,Q20,Q21,Q22,Q23,Q24},
       legend columns=-1,
       legend style={at={(0.52,1)},anchor=south west,draw=none,fill=none,font=\small},
         legend image code/.code={%
                    \draw[#1, draw=black] (0cm,-0.1cm) rectangle (0.5cm,0.2cm);
                },
        area legend,
       xmin=12.5,
       xmax=24.5,
      ymax = 80,
      ymin=0,
]
        
         
         \addplot[ybar, pattern color=cyan!150, pattern=crosshatch dots] table[
           x expr=\thisrowno{0}, 
           y expr=\thisrowno{1}]{\pg};
        \addplot[ybar, pattern color=blue!150, pattern=north east lines, ] table[
           x expr=\thisrowno{0}, 
           y expr=\thisrowno{1}]{\nj};
        \addplot[ybar, pattern color=red!150, pattern=grid] table[
           x expr=\thisrowno{0}, 
           y expr=\thisrowno{1}]{\vj};
          
        \addplot[ybar, pattern color=green!150, pattern=north west lines] table[
           x expr=\thisrowno{0}, 
           y expr=\thisrowno{1}]{\eh};
        \addplot[ybar,fill=black!20] table[
           x expr=\thisrowno{0}, 
           y expr=\thisrowno{1}]{\xjoin};

    \end{groupplot}

 
\end{tikzpicture}

%% file: evaluation/intermediate.tex
\begin{tikzpicture}

\pgfplotstableread{		
x	y	ac
1	123335	123335
2	91738	215073
3	39005	254078
4	34648	288726
5	89636	378362
6	7685	386047
}{\njone}		
		
\pgfplotstableread{		
x	y	ac
1	123335	123335
2	91738	215073
3	39005	254078
4	34648	288726
5	7685	296411
}{\vjone}		
		
\pgfplotstableread{		
x	y	ac
1	84866	84866
2	27800	112666
3	13589	126255
4	4239	130494
5	7575	138069
6	7380	145449
}{\xjoinone}		

\pgfplotstableread{		
x	a	b	c	d	e	f   g
2	315197	155744	155744	154541	124363	153270	116824
3	123335	91738	39005	34648	89636	7685    0
4	123335	91738	39005	34648	7685	0   0
5   84866   27800   13589   20515   32858   64995   7685
6	84866	27800	13589	4239	7575	7380    0
}{\onebar}	
		
\pgfplotstableread{		
x	y	ac
1	123335	123335
2	91738	215073
3	39005	254078
4	34648	288726
5	89636	378362
6	10000000	10378362
}{\njtwo}		
		
\pgfplotstableread{		
x	y	ac
1	123335	123335
2	91738	215073
3	39005	254078
4	34648	288726
5	7685	296411
6	4604	301015
}{\vjtwo}		
		
\pgfplotstableread{		
x	y	ac
1	41807	41807
2	13587	55394
3	11586	66980
4	2498	69478
5	4541	74019
6	4605	78624
}{\xjointwo}		

\pgfplotstableread{		
x	a	b	c	d	e	f
3	123335	91738	39005	34648	89636	10000000
4	123335	91738	39005	34648	7685	4604
5	41807	13587	11586	2498	4541	4605
}{\twobar}

\pgfplotstableread{		
x	y	ac
1 153270	153270
2 124363	277633
3 155744	589121
4 154541	898203
5 75431	1058859
6 5816204	6875063
7 621	6875684
}{\pgthree}			

\pgfplotstableread{		
x	a	b	c	d	e	f	g
2	153270	124363	155744	154541	75431	5816204	621
3	123335	91738	39005	34648	89636	4636604	88
4	123335	91738	39005	34648	7685	88	0
5   51734	122426	83043	13588	32901	64995	88
6	41807	13580	20336	4236	82	89	0
}{\threebar}

\pgfplotstableread{		
x	y	ac
1	123335	123335
2	91738	215073
3	39005	254078
4	34648	288726
5	89636	378362
6	4636604	5014966
7	88	5015054
}{\njthree}

\pgfplotstableread{		
x	y	ac
1	123335	123335
2	91738	215073
3	39005	254078
4	34648	288726
5	7685	296411
6	88	296499
}{\vjthree}

\pgfplotstableread{		
x	y	ac
1	41807	41807
2	13580	55387
3	20336	75723
4	4236	79959
5   82  80041
6	89	80130
}{\xjointhree}		

\pgfplotstableread{		
x	a	b	c	d	e	f   g 
2	21750	21750	82151	6591894	31765950	91867   0 
3	82151	917913	25315365	4353956	94345	0   0 
4	917913	25315365	4353956	100000000	0	0   0 
5   21751	79826	80855	2338304	63702	890385	4353957
6	21750	21411	78112	79122	94345	0   0
}{\fourbar}		

\pgfplotstableread{		
x	y	ac
1	0	0
}{\pgonefour}	

\pgfplotstableread{		
x	y	ac
1	82151	82151
2	917913	1000064
3	25315365	26315429
4	4353956	30669385
5	94345	30763730
}{\njonefour}		
		
\pgfplotstableread{		
x	y	ac
1	917913	917913
2	25315365	26233278
3	4353956	30587234
4	100000000	150000000
}{\vjonefour}		
		
\pgfplotstableread{		
x	y	ac
1	21750	21750
2	21411	43161
3	78112	121273
4	79122	200395
5	94345	294740
}{\xjoinonefour}

\pgfplotstableread{		
x	y	ac
1	21750	21750
2	21750	43500
3	82151	125651
4	6591894	6717545
5	31765950	38483495
6	91867	38575362
}{\pgfour}	

\pgfplotstableread{		
x	y	ac
1	82151	82151
2	917913	1000064
3	25315365	26315429
4	4353956	30669385
5	94345	30763730
}{\njfour}		
		
\pgfplotstableread{		
x	y	ac
1	917913	917913
2	25315365	26233278
3	4353956	30587234
4	100000000	150000000
}{\vjfour}		
		
\pgfplotstableread{		
x	y	ac
1	21750	21750
2	21411	43161
3	78112	121273
4	79122	200395
5	94345	294740
}{\xjoinfour}

\pgfplotstableread{		
x	y	ac
1	82151	82151
2	30023	112174
3	78633	190807
4	917913	1108720
5	25315365	26424085
6	4353956	30778041
7	656	30778697
}{\njfive}		
		
\pgfplotstableread{		
x	y	ac
1	917913	917913
2	25315365	26233278
3	4353956	30587234
4	100000000	150000000
}{\vjfive}		
		
\pgfplotstableread{		
x	y	ac
1	13174	13174
2	48339	61513
3	47297	108810
4	652	109462
5	656	110118
}{\xjoinfive}		

\pgfplotstableread{		
x	a	b	c	d	e	f g
2	82151	21750	6591894	31765950	403 0   0
3	82151	30023	78633	917913	25315365	4353956	656
4	917913	25315365	4353956	100000000	0	0	0
5   21751	79826	80855	2338304	63180	889971	4353880	
6	13174	48339	47297	652	656	0	0	
}{\fivebar}	
\pgfplotstableread{		
x	y	ac
1	20946	20946
2	917913	938859
3	25315365	26254224
4	4353956	30608180
5	3965950	34574130
6	3965950	38540080
}{\njsix}		
		
\pgfplotstableread{		
x	y	ac
1	917913	917913
2	25315365	26233278
3	4353956	30587234
4	100000000	150000000
}{\vjsix}		
		
\pgfplotstableread{		
x	y	ac
1	13174	13174
2	12959	26133
3	20401	46534
4	20475	67009
5	591690	658699
6	3965950	4624649
}{\xjoinsix}		

\pgfplotstableread{		
x	a	b	c	d	e	f
2   82151	21750	21750	6591894	31765950	2486184
3	20946	917913	25315365	4353956	3965950	3965950
4	917913	25315365	4353956	100000000	0	0
5    13175	20873	63702	890385	4353957	3947325
6	13174	12959	20401	20475	591690	3965950
}{\sixbar}

\pgfplotstableread{		
x	a	b	c	d	e	f   g
2	315197	155744	155744	154541	124363	153270	116824
3	125997	123335	91738	34648	6328	89636	1368
4	125997	123335	91738	39005	34648	6328	1368
5   51734	122426	84858	27800	8614	32901	64995
6	122426	84858	5577	27800	1368	0	0
}{\onefourbar}	

\pgfplotstableread{		
x	a	b	c	d	e	f   g
3	125997	123335	91738	39005	34648	6328	10000000
4	125997	123335	91738	39005	34648	1368	792
5	122426	84858	5577	5577	2728	482	0
}{\onefivebar}	

\pgfplotstableread{		
x	a	b	c	d	e	f g
3	125997	123335	91738	39005	6328	89636	1368
4	125997	123335	91738	39005	34648	6328	1368
5   13175	20873	63702	890385	4353957	3947325	0
6	122426	84858	5577	5577	827	13	0
}{\onesixbar}

\pgfplotstableread{		
x	y	ac
1	999993	999993
2	3332985	4332978
3	11119727	15452705
4	37044044	52496749
}{\njseven}		
		
\pgfplotstableread{		
x	y	ac
1	999993	999993
2	3332985	4332978
3	11119727	15452705
4	37044044	52496749
}{\vjseven}		
		
\pgfplotstableread{		
x	y	ac
1	198641	198641
2	993033	1191674
3	4930346	6122020
4	37044044	43166064
}{\xjoinseven}		

\pgfplotstableread{		
x	a	b	c	d   e   f   g
2	3332858	1000000	3334961	37044586    0   0   0
3	999993	3332985	11119727	37044044    0   0   0
4	999993	3332985	11119727	37044044    0   0   0
5   279064	964470	3334941	279055	964776	3332829	37044044
6	198641	993033	4930346	37044044    0   0   0
}{\sevenbar}	
	
\pgfplotstableread{		
x	y	ac
1	1000000	1000000
2	10	1000010
3	38	1000048
4	92	1000140
}{\njeight}		
		
\pgfplotstableread{		
x	y	ac
1	1000000	1000000
2	10	1000010
3	38	1000048
4	92	1000140
}{\vjeight}		
		
\pgfplotstableread{		
x	y	ac
1	198641	198641
2	193067	391708
3	18	391726
4	84	391810
5	92	391902
}{\xjoineight}		
		
\pgfplotstableread{		
x	a	b	c	d	e 
2   3332858	1000000	3334961	92  0  
3	1000000	10	38	92	0 
4	1000000	10	38	92	0
5   279064	964470	279041	0   0
6	198641	193067	18	84	92
}{\eightbar}			
		
\pgfplotstableread{		
x	y	ac
1	1000000	1000000
2	13	1000013
3	43	1000056
4	137	1000193
}{\njnine}		
		
\pgfplotstableread{		
x	y	ac
1	1000000	1000000
2	13	1000013
3	43	1000056
4	137	1000193
}{\vjnine}		
		
\pgfplotstableread{		
x	y	ac
1	279067	279067
2	264056	543123
3	13	543136
4	44	543180
5	137	543317
}{\xjoinnine}		

\pgfplotstableread{		
x	a	b	c	d	e
2	3332858	1000000	3334961	137 0
3	1000000	13	43	137	0
4	1000000	13	43	137	0
5	279067	264056	13	44	137
}{\ninebar}	
		
\pgfplotstableread{		
x	y	ac
1	229166	229166
2	229166	458332
3	104	458436
4	1017	459453
5	1017	460470
6	1143	461613
}{\njten}		
		
\pgfplotstableread{		
x	y	ac
1	229166	229166
2	2289575	2518741
3	1143	2519884
4	1143	2521027
5	1143	2522170
6	1143	2523313
}{\vjten}		
		
\pgfplotstableread{		
x	y	ac
1	136481	136481
2	229142	365623
3	116	365739
4	1143	366882
}{\xjointen}		

\pgfplotstableread{		
x	a	b	c	d	e	f   G
2	229166	229166	1017	1100    0   0   0
3	229166	229166	104	1017	1017	1143    0
4	229166	2289575	1143	1143	1143	1143    0
5   136481	229167	229167	199980	105	1018	1143
6	136481	229142	116	1143	0	0   0
}{\tenbar}	
		
\pgfplotstableread{		
x	y	ac
1	50522	50522
2	229166	279688
3	229166	508854
4	1999957	2508811
5	104	2508915
6	1017	2509932
7	1	2509933
}{\njeleven}		
		
\pgfplotstableread{		
x	y	ac
1	50522	50522
2	229166	279688
3	229166	508854
4	229166	738020
5	20	738040
6	1	738041
7	1	738042
}{\vjeleven}		
		
\pgfplotstableread{		
x	y	ac
1	44501	44501
2	50523	95024
3	229155	324179
4	1	324180
}{\xjoineleven}		
		
\pgfplotstableread{		
x	a	b	c	d	e	f   g
2	229166	229166	50522	1017    0   0   0
3	50522	229166	229166	1999957	104	1017	1
4	50522	229166	229166	229166	20	1	1
5   44501	50523	229167	44501	50521	1	0
6	44501	50523	229155	1	0	0	0	
}{\elevenbar}			
		
\pgfplotstableread{		
x	y	ac
1	50522	50522
2	229166	279688
3	1999957	2279645
4	104	2279749
5	1017	2280766
6	1	2280767
}{\njtwelve}		
		
\pgfplotstableread{		
x	y	ac
1	50522	50522
2	229166	279688
3	229166	508854
4	20	508874
5	1	508875
6	1	508876
}{\vjtwelve}		
		
\pgfplotstableread{		
x	y	ac
1	44501	44501
2	50523	95024
3	229140	324164
4	1	324165
}{\xjointwelve}		

\pgfplotstableread{		
x	a	b	c	d	e	f
2	229166	229166	50522	1017    0   0
3	50522	229166	1999957	104	1017	1
4	50522	229166	229166	20	1	1
5   44501	50523	229167	199985	105	1018
6	44501	50523	229140	1	0	0	
}{\twelvebar}

\pgfplotstableread{		
x	a	b	c	d	e	f   g
2   3332858	1000000	3334961	92  0   0   0
3	82151	82151	78633	917913	25315365    4353956	655
4	82151	82151	82151	78633	12587	2541	655
5   21751	79826	80855	2338304	63180	889971	4353880	
6	331	13174	12961	46705	652	0	0
}{\twotenbar}	

\pgfplotstableread{		
x	a	b	c	d	e	f   g
2   3332858	1000000	3334961	92  0   0   0
3	82151	82151	78633	917913	25315365    4353956	655
4	82151	82151	82151	78633	12587	2541	655
5	331	13174	12961	46705	652	0	0
}{\twoelevenbar}	

\pgfplotstableread{		
x	a	b	c	d	e	f   g
2   3332858	1000000	3334961	92  0   0   0
3	82151	82151	78633	917913	25315365    4353956	655
4	82151	82151	82151	78633	12587	2541	655
5	331	13174	12961	46705	652	0	0
}{\twotwelvebar}	

\pgfplotstableread{		
x	a	b	c	d	e	f   g
2	229166	229166	229166	229166  1017   1100 0
3	229166	229166	2000000	104	1017	1143    0
4	229166	229166	229166	2289575	115	1143    0
5   136481	229167	229167	199980	105	1018	1143
6	229167	229142	116	116	0	0   0
}{\nineteenbar}	

\pgfplotstableread{		
x	a	b	c	d	e	f
2	50522	229166	50522	229166    229166   1017 
3	50522	229166	229166	2000000	104	1017
4	50522	229166	229166	229166	9	0
5   44501	50523	229167	44501	50521	1
6	50523	50521	100	100	0	0
}{\twentybar}	

\pgfplotstableread{		
x	a	b	c	d	e	f   g
2	50522	229166	229166	229166    229166   229166   1017
3	50522	229166	1105234	229166	1999957	104	1017
4	50522	229166	229166	1105234	229166	20	0
5   44501	50523	229165	1105235	44501	1   0
6	50523	50521	100	100	0	0	0
}{\twentyonebar}

\pgfplotstableread{		
x	a	b	c
2   573146  0    573146
3	573146	10000000	10573146
4	573146	10000000	10573146
5	573146	0	573146
}{\twentytwobar}

\pgfplotstableread{		
x	a	b	c
2   573146  10000000     10573146
3	573146	10000000	10573146
4	573146	10000000	10573146
5	573146	0	573146
}{\twentythreebar}

\pgfplotstableread{		
x	a	b	c	d   e
2   573146  4124302 3436935    0   8134383 
3	573146	573146	573146	0	1719438
4	573146	573146	573146	573146	2292584
5   471	544942	1	1   0
6	573146	0	0	0	573146
}{\twentyfourbar}	
         
\begin{groupplot}[
        width=2.5cm, 
		height=3.0cm, 
		label style={font=\tiny},
 		ticklabel style={font=\tiny},
		title style={at={(0.9,1.07)},anchor=north,draw=none,fill=none,font=\scriptsize},
        group style={group name=plots, group size=8 by 2, horizontal sep=1.5em, vertical sep=1.4em},
        ]

\nextgroupplot[
        ybar stacked, bar width=0.3cm, 
        title=Q1,
        bar width=0.11cm,
        xtick={2,3,4,5,6},
        ylabel=Intermediate~~result~~size,
        ylabel style={at={(-1.0,-0.3)},anchor=north,font=\small},
        xticklabel style={font=\tiny},
        xticklabel=\empty,
       legend columns=-1,
       ymax=1000000,
       legend style={at={(0,1)},anchor=south west,draw=none,fill=none,font=\footnotesize},
        ]
        \addplot table [x=x, y=a] {\onebar};
        \addplot table [x=x, y=b] {\onebar};
        \addplot table [x=x, y=c] {\onebar};
        \addplot table [x=x, y=d] {\onebar}; 
        \addplot table [x=x, y=e] {\onebar};
        \addplot table [x=x, y=f] {\onebar};
        \addplot table [x=x, y=g] {\onebar};


\nextgroupplot[
        ybar stacked, bar width=0.3cm, 
        title=Q3,
        bar width=0.11cm,bar width=0.11cm,xtick={2,...,10},
        xticklabel style={font=\tiny},
        xticklabel=\empty,
       legend columns=-1,
       ymax=1000000,
       ymin=0,
       legend style={at={(10,1.5)},
       draw=none,fill=none,font=\footnotesize},
        ]
        \addplot table [x=x, y=a] {\threebar};\addlegendentry{\textit{JS1}};
        \addplot table [x=x, y=b] {\threebar};\addlegendentry{\textit{JS2}};
        \addplot table [x=x, y=c] {\threebar};\addlegendentry{\textit{JS3}};
        \addplot table [x=x, y=d] {\threebar};\addlegendentry{\textit{JS4}};  
        \addplot table [x=x, y=e] {\threebar};\addlegendentry{\textit{JS5}};
        \addplot table [x=x, y=f] {\threebar};\addlegendentry{\textit{JS6}};
        \addplot[fill,black] table [x=x, y=g] {\threebar};\addlegendentry{\textit{JS7}};  

\nextgroupplot[
        ybar stacked, bar width=0.3cm, 
        title=Q4,
        bar width=0.11cm,xtick={2,...,10},
        xticklabel style={font=\tiny},
        xticklabel=\empty,
       legend columns=4,
       ymax=500000,
       ymin=0,
       legend style={at={(0,1)},anchor=south west,draw=none,fill=none,font=\footnotesize},
        ]
        \addplot table [x=x, y=a] {\onefourbar};
        \addplot table [x=x, y=b] {\onefourbar};
        \addplot table [x=x, y=c] {\onefourbar};
        \addplot table [x=x, y=d] {\onefourbar};  
        \addplot table [x=x, y=e] {\onefourbar};
        \addplot table [x=x, y=f] {\onefourbar};
        \addplot[fill,black] table [x=x, y=g] {\onefourbar};



\nextgroupplot[
        ybar stacked, bar width=0.3cm, 
        title=Q7,
        bar width=0.11cm,xtick={2,...,10},
        xticklabel style={font=\tiny},
        xticklabel=\empty,
       legend columns=4,
       ymax=5000000,
       ymin=0,
       legend style={at={(0,1)},anchor=south west,draw=none,fill=none,font=\footnotesize},
        ]
        \addplot table [x=x, y=a] {\fourbar};
        \addplot table [x=x, y=b] {\fourbar};
        \addplot table [x=x, y=c] {\fourbar};
        \addplot table [x=x, y=d] {\fourbar};  
        \addplot table [x=x, y=e] {\fourbar};
        \addplot table [x=x, y=f] {\fourbar};
        \addplot[fill,black] table [x=x, y=g] {\fourbar};

\nextgroupplot[
        ybar stacked, bar width=0.3cm, 
        title=Q8,
        bar width=0.11cm,xtick={2,...,10},
        xticklabel style={font=\tiny},
        xticklabel=\empty,
       legend columns=4,
       ymax=1000000,
       ymin=0,
       legend style={at={(0,1)},anchor=south west,draw=none,fill=none,font=\footnotesize},
        ]
        \addplot table [x=x, y=a] {\fivebar};
        \addplot table [x=x, y=b] {\fivebar};
        \addplot table [x=x, y=c] {\fivebar};
        \addplot table [x=x, y=d] {\fivebar};  
        \addplot table [x=x, y=e] {\fivebar};
        \addplot table [x=x, y=f] {\fivebar};
        \addplot[fill,black] table [x=x, y=g] {\fivebar};

\nextgroupplot[
        ybar stacked, bar width=0.3cm, 
        title=Q9,
        bar width=0.11cm,xtick={2,...,10},
        xticklabel style={font=\tiny},
        xticklabel=\empty,
       legend columns=4,
       ymax=50000000,
       ymin=0,
       legend style={at={(0,1)},anchor=south west,draw=none,fill=none,font=\footnotesize},
        ]
        \addplot table [x=x, y=a] {\sixbar};
        \addplot table [x=x, y=b] {\sixbar};
        \addplot table [x=x, y=c] {\sixbar};
        \addplot table [x=x, y=d] {\sixbar};  
        \addplot table [x=x, y=e] {\sixbar};
        \addplot table [x=x, y=f] {\sixbar};

\nextgroupplot[
        ybar stacked, bar width=0.3cm, 
        title=Q10,
        bar width=0.11cm,xtick={2,...,10},
        xticklabel style={font=\tiny},
        xticklabel=\empty,
       legend columns=4,
       ymax=1000000,
       ymin=0,
       legend style={at={(0,1)},anchor=south west,draw=none,fill=none,font=\footnotesize},
        ]
        \addplot table [x=x, y=a] {\twotenbar};
        \addplot table [x=x, y=b] {\twotenbar};
        \addplot table [x=x, y=c] {\twotenbar};
        \addplot table [x=x, y=d] {\twotenbar};  
        \addplot table [x=x, y=e] {\twotenbar};
        \addplot table [x=x, y=f] {\twotenbar};
        \addplot[fill,black] table [x=x, y=g] {\twotenbar};

        

\nextgroupplot[
        ybar stacked, bar width=0.3cm, 
        title=Q13,
        bar width=0.11cm,xtick={2,...,10},
        xticklabel style={font=\tiny},
        xticklabel=\empty,
       legend columns=4,
       ymax=100000000,
       ymin=0,
       legend style={at={(0,1)},anchor=south west,draw=none,fill=none,font=\footnotesize},
        ]
        \addplot table [x=x, y=a] {\sevenbar};
        \addplot table [x=x, y=b] {\sevenbar};
        \addplot table [x=x, y=c] {\sevenbar};
        \addplot table [x=x, y=d] {\sevenbar};
        \addplot table [x=x, y=e] {\sevenbar};
        \addplot table [x=x, y=f] {\sevenbar};
        \addplot[fill,black] table [x=x, y=g] {\sevenbar};
        
\nextgroupplot[
        ybar stacked, bar width=0.3cm, 
        title=Q14,
        bar width=0.11cm,xtick={2,...,10},
        xticklabel style={font=\tiny},
        xticklabels={PG,SJ,VJ,EH,CM},  
        xticklabel style={rotate=90},
       legend columns=4,
       ymax=10000000,
       ymin=0,
       legend style={at={(0,1)},anchor=south west,draw=none,fill=none,font=\footnotesize},
        ]
        \addplot table [x=x, y=a] {\eightbar};
        \addplot table [x=x, y=b] {\eightbar};
        \addplot table [x=x, y=c] {\eightbar};
        \addplot table [x=x, y=d] {\eightbar};  
        \addplot table [x=x, y=e] {\eightbar};


\nextgroupplot[
        ybar stacked, bar width=0.3cm, 
        title=Q16,
        bar width=0.11cm,xtick={2,...,10},
        xticklabels={PG,SJ,VJ,EH,CM},  
        xticklabel style={rotate=90},
        xticklabel style={font=\tiny},
       legend columns=4,
       ymax=1000000,
       ymin=0,
       legend style={at={(0,1)},anchor=south west,draw=none,fill=none,font=\footnotesize},
        ]
        \addplot table [x=x, y=a] {\tenbar};
        \addplot table [x=x, y=b] {\tenbar};
        \addplot table [x=x, y=c] {\tenbar};
        \addplot table [x=x, y=d] {\tenbar};  
        \addplot table [x=x, y=e] {\tenbar};
        \addplot table [x=x, y=f] {\tenbar};
        
        \nextgroupplot[
        ybar stacked, bar width=0.3cm, 
        title=Q17,
        bar width=0.11cm,xtick={2,...,10},
        xticklabels={PG,SJ,VJ,EH,CM},  
        xticklabel style={rotate=90},
        xticklabel style={font=\tiny},
       legend columns=4,
       ymax=1000000,
       ymin=0,
       legend style={at={(0,1)},anchor=south west,draw=none,fill=none,font=\footnotesize},
        ]
        \addplot table [x=x, y=a] {\elevenbar};
        \addplot table [x=x, y=b] {\elevenbar};
        \addplot table [x=x, y=c] {\elevenbar};
        \addplot table [x=x, y=d] {\elevenbar};  
        \addplot table [x=x, y=e] {\elevenbar};
        \addplot table [x=x, y=f] {\elevenbar};
        
        \nextgroupplot[
        ybar stacked, bar width=0.3cm, 
        title=Q18,
        bar width=0.11cm,xtick={2,...,10},
        xticklabels={PG,SJ,VJ,EH,CM},  
        xticklabel style={rotate=90},
        xticklabel style={font=\tiny},
       legend columns=4,
       ymax=1000000,
       ymin=0,
       legend style={at={(0,1)},anchor=south west,draw=none,fill=none,font=\footnotesize},
        ]
        \addplot table [x=x, y=a] {\twelvebar};
        \addplot table [x=x, y=b] {\twelvebar};
        \addplot table [x=x, y=c] {\twelvebar};
        \addplot table [x=x, y=d] {\twelvebar};  
        \addplot table [x=x, y=e] {\twelvebar};
        \addplot table [x=x, y=f] {\twelvebar};

\nextgroupplot[
        ybar stacked, bar width=0.3cm, 
        title=Q19,
        bar width=0.11cm,xtick={2,...,10},
        xticklabels={PG,SJ,VJ,EH,CM},  
        xticklabel style={rotate=90}, 
        xticklabel style={font=\tiny},
       legend columns=4,
       ymax=1000000,
       ymin=0,
       legend style={at={(0,1)},anchor=south west,draw=none,fill=none,font=\footnotesize},
        ]
        \addplot table [x=x, y=a] {\nineteenbar};
        \addplot table [x=x, y=b] {\nineteenbar};
        \addplot table [x=x, y=c] {\nineteenbar};
        \addplot table [x=x, y=d] {\nineteenbar};  
        \addplot table [x=x, y=e] {\nineteenbar};
        \addplot table [x=x, y=f] {\nineteenbar};
        \addplot[fill,black] table [x=x, y=g] {\nineteenbar};
        
        \nextgroupplot[
        ybar stacked, bar width=0.3cm, 
        title=Q20,
        bar width=0.11cm,xtick={2,...,10},
        xticklabels={PG,SJ,VJ,EH,CM},  
        xticklabel style={rotate=90},
        xticklabel style={font=\tiny},
       legend columns=4,
       ymax=1000000,
       ymin=0,
       legend style={at={(0,1)},anchor=south west,draw=none,fill=none,font=\footnotesize},
        ]
        \addplot table [x=x, y=a] {\twentybar};
        \addplot table [x=x, y=b] {\twentybar};
        \addplot table [x=x, y=c] {\twentybar};
        \addplot table [x=x, y=d] {\twentybar};  
        \addplot table [x=x, y=e] {\twentybar};
        \addplot table [x=x, y=f] {\twentybar};
        
\nextgroupplot[
        ybar stacked, bar width=0.3cm, 
        title=Q21,
        bar width=0.11cm,xtick={2,...,10},
        xticklabels={PG,SJ,VJ,EH,CM},  
        xticklabel style={rotate=90},
        xticklabel style={font=\tiny},
       legend columns=4,
       ymax=1000000,
       ymin=0,
       legend style={at={(0,1)},anchor=south west,draw=none,fill=none,font=\footnotesize},
        ]
        \addplot table [x=x, y=a] {\twentyonebar};
        \addplot table [x=x, y=b] {\twentyonebar};
        \addplot table [x=x, y=c] {\twentyonebar};
        \addplot table [x=x, y=d] {\twentyonebar};  
        \addplot table [x=x, y=e] {\twentyonebar};
        \addplot table [x=x, y=f] {\twentyonebar}; 
        \addplot[fill,black] table [x=x, y=g] {\twentyonebar};



\nextgroupplot[
        ybar stacked, bar width=0.3cm, 
        title=Q24,
        bar width=0.11cm,xtick={2,...,10},
        xticklabels={PG,SJ,VJ,EH,CM},  xticklabel style={font=\tiny},
       legend columns=4,
       xticklabel style={rotate=90},
       ymax=5000000,
       ymin=0,
       legend style={at={(0,1)},anchor=south west,draw=none,fill=none,font=\footnotesize},
        ]
        \addplot table [x=x, y=a] {\twentyfourbar};
        \addplot table [x=x, y=b] {\twentyfourbar};
        \addplot table [x=x, y=c] {\twentyfourbar};
        \addplot table [x=x, y=d] {\twentyfourbar};  

\end{groupplot} 
\end{tikzpicture}

%% file: section/discussion.tex

%% file: section/relatedwork.tex
\section{Related work}\label{sec:related_work}



\noindent \textbf{Worst-case size bounds and optimal algorithms} \indent  
Recently, Grohe and Marx \cite{10.1145/2636918} and Atserias, Grohe, and Marx  \cite{DBLP:conf/focs/AtseriasGM08} estimated size bounds for conjunctive joins using the fractional edge cover. That allows us to compute the worst-case size bound by linear programming. Based on this worst-case bound, several worst-case optimal algorithms have been proposed (e.g. \textit{NPRR} \cite{DBLP:journals/corr/abs-1203-1952}, \textit{LeapFrog} \cite{veldhuizen2012leapfrog}, \textit{Joen} \cite{ciucanu2015worst}). Ngo et al. \cite{DBLP:journals/corr/abs-1203-1952} constructed the first algorithm whose running time is worst-case optimal for all natural join queries. Veldhuizen \cite{veldhuizen2012leapfrog} proposed an optimal algorithm called \textit{LeapFrog} which is efficient in practice to implement. Ciucanu et al. \cite{ciucanu2015worst}  proposed an optimal algorithm \textit{Joen} which joins each attribute at a time via an improved tree representation. Besides, there exist research works on applying functional dependencies (FDs) for size bound estimation. The initiated study with FDs is from Gottlob, Lee, Valiant, and Valiant (GLVV) \cite{DBLP:journals/jacm/GottlobLVV12}, which introduces an upper bound called GLVV-bound based on a solution of a linear program on polymatroids. The follow-up study by Gogacz et al. \cite{gogacz2015entropy} provided a precise characterization of the worst-case bound with information theoretic entropy. Khamis et al. \cite{DBLP:conf/pods/KhamisNS16} provided a worst-case optimal algorithm for any query where the GLVV-bound is tight. See an excellent survey on the development of worst-case bound theory \cite{DBLP:journals/sigmod/NgoRR13}.

\noindent \textbf{Multi-model data management} \indent  As more businesses realized that data, in all forms and sizes, are critical to make the best possible decisions, we see a continuing growth of demands to manage and process massive volumes of different types of data \cite{DBLP:conf/edbt/LuH17}. The data are represented in various models and formats: structured, semi-structured, and unstructured. A traditional database typically handles only one data model. It is promising to develop a multi-model database to manage and process multiple data models against a single unified backend while meeting the increasing requirements for scalability and performance \cite{DBLP:conf/edbt/LuH17,DBLP:journals/corr/LuLXZ16}. Yet, it is challenging to process and optimize cross-model queries. 

Previous work applied naive or no optimizations on (relational and tree) CMCQs. There exist two kinds of solutions. The first is to use one query to retrieve the result from the system without changing the nature of the model \cite{DBLP:conf/fnc/NassiriMH18,zhangunibench}. The second is to encode and retrieve the tree data into a relational engine \cite{DBLP:conf/sigmod/AbergerTOR16,DBLP:journals/jsw/BousalemC15,DBLP:journals/kbs/QtaishA16,DBLP:conf/ACISicis/ZhuYFS17}. Even though the second solution accelerates twig matching, they both may suffer from generating large, unnecessary intermediate results. These solutions or optimizations did not consider cross-model worst-case optimality. Some advances are already in development to process graph patterns \cite{DBLP:conf/sigmod/NguyenABKNRR14,DBLP:conf/semweb/HoganRRS19,DBLP:conf/sigmod/AbergerTOR16}. In contrast to previous work, this paper initiates the study on the worst-case bound  for cross-model conjunctive queries with both relation and tree structure data.

\noindent \textbf{Join order} \indent 
%
In this paper, we do not focus on more complex query plan optimization. A better query plan~\cite{gottlob2005hypertree,DBLP:conf/pods/FischlGP18} may lead to a better bound for some instances~\cite{DBLP:conf/sigmod/AbergerTOR16} by combining the worst-case optimal algorithm and non-cyclic join optimal algorithm (i.e. Yannakakis \cite{DBLP:conf/vldb/Yannakakis81}). We leave this as the future work to continue optimizing CMCQs.

%% file: section/conclusion.tex
\begin{sidewaystable}
\tiny
\centering
\caption{Dataset statistics and designed queries~($m$=$10^6$, $k$=$10^3$).}
\label{tab:dataset}

\input{table/dataset}
\end{sidewaystable}

\section{Conclusion and future work}\label{sec:conclusion}
In this paper, we studied the problems to find the worst-case size bound and optimal algorithm for cross-model conjunctive queries with relation and tree structured data. 
We provide the optimized algorithm, i.e. CMJoin, to compute the worst-case bound and the worst-case optimal algorithm for cross-model joins.  Our experimental results demonstrate the superiority of proposal algorithms against state-of-the-art systems and algorithms in terms of efficiency, scalability, and intermediate cost. Exciting follow-ups will focus on adding graph structured data into our problem setting and designing a more general cross-model  algorithm involving three data models, i.e. relation, tree and graph.

\section{Acknowledgment}\label{sec:acknowledgment}
This paper is partially supported by Finnish Academy Project 310321 and Oracle ERO gift funding.

%% file: table/dataset.tex
\setlength{\tabcolsep}{2pt}
\setlength\extrarowheight{4pt}
\begin{tabular}{llcl p{2.4in} cc}
\toprule
\textbf{Dataset}                   & \textbf{Statistics}                       & \textbf{Query}            & \textbf{Relational table}                                                                                                                                       & \textbf{XML or JSON path query}                                           & \textbf{LP}& \textbf{\#Result} \\ \hline
                          
\multirow{6}{*}{\begin{tabular}[c]{@{}l@{}} D1:TreeBank\cite{xue2005penn_treebank}\\ (Linguistic data)\end{tabular}} & \multirow{6}{*}{\begin{tabular}[c]{@{}l@{}}Zipfian\\  Tables:~1$m$ rows\\ XML:~2.4$m$ nodes\end{tabular} } & Q1 & 
R1(NP,VP)    & \multirow{3}{*}{S[NP]/VP//PP[IN]//NNP}                         & $N^3$    &7.6$k$   \\ \cline{3-4} \cline{6-7} 
                          &                                             & Q2  &   R1(NP,VP)~R2(NP,PP)                                                                                                                                            &                           & $N^3$   &4.6$k$   \\ \cline{3-4} \cline{6-7}  
                          &                                            & Q3  &   R1(NP,VP)~R3(NP,NNP)                                                                                                                                                &                         & $N^3$   &<0.1$k$   \\ \cline{3-7}
&  & Q4 & 
R1(NP,VP)    & \multirow{3}{*}{\begin{tabular}[c]{@{}l@{}}S[NP]/VP//PP[IN]//NNP \\ S/VP/PP/IN \end{tabular}}                         & $N^3$    &1.4$k$   \\ \cline{3-4} \cline{6-7} 
                          &                                             & Q5  &   R1(NP,VP)~R2(NP,PP)                                                                                                                                            &                           & $N^2$   &0.8$k$   \\ \cline{3-4} \cline{6-7}  
                          &                                            & Q6  &   R1(NP,VP)~R3(NP,NNP)                                                                                                                                                &                         & $N^3$   & <0.1$k$  \\ \hline
\multirow{6}{*}{\begin{tabular}[c]{@{}l@{}} D2:Xmark\cite{DBLP:conf/vldb/SchmidtWKCMB02}\\ (Auction data)\end{tabular}}            
& \multirow{6}{*}{\begin{tabular}[c]{@{}l@{}}Normal\\  Tables:~1$m$ rows\\ XML:~1.6$m$ nodes\end{tabular}}         & Q7   & \multirow{6}{*}{\begin{tabular}[c]{@{}l@{}}R1(incategory,quantity,email)\\ R2(item,incategory,email)\\ R3(item,quantity,email)\end{tabular}}    & T7=Item{[}incategory{]}/quantity                  & $N^2$  &91$k$   \\ \cline{3-3} \cline{5-7}
                          &                                     & Q8          &                                                                                                                                                 & T8=Item{[}incategory{]}{[}localtion{]}{[}quantity{]}//email   & $N^3$ &0.4$k$    \\ \cline{3-3} \cline{5-7} 
                          &                                     & Q9         &                                                                                                                                                  & T9=Item{[}location{]}//email                       & $N^3$    &2.4$m$  \\ \cline{3-3} \cline{5-7} 

&   & Q10   &    & T7, T8                  & $N^3$  &0.7$k$   \\ \cline{3-3} \cline{5-7}
                          &                                     & Q11          &                                                                                                                                                 & T7, T9   & $N^3$ &0.7$k$    \\ \cline{3-3} \cline{5-7} 
                          &                                     & Q12         &                                                                                                                                                  & T7, T8, T9                       & $N^3$    &0.7$k$  \\ \hline

\multirow{3}{*}{\begin{tabular}[c]{@{}l@{}} D3:UniBench\cite{zhangunibench}\\ (E-commerce)\end{tabular}}  & \multirow{3}{*}{\begin{tabular}[c]{@{}l@{}}Uniform\\  Tables:~1$m$ rows\\ JSON:~2$m$-4$m$ nodes\end{tabular}}   & Q13    & \multirow{3}{*}{\begin{tabular}[c]{@{}l@{}}R1(asin,productID,orderID)\\ R2(personID,lastname) \\ R3(productID,product\_info)\end{tabular}}    & \$.[orderID,personID]                     & $N^3$   &37.0$m$   \\ \cline{3-3} \cline{5-7} 
                          &                                           & Q14     &                                                                                                                                                & \$.[orderID,personID,orderline[productID]] & $N^3$  &<0.1$k$    \\ \cline{3-3} \cline{5-7} 
                          &                                            & Q15   &                                                                                                                                                 & \$.[personID,orderline[productID, asin]]                    & $N^3$  &0.1$k$    \\ \hline

\multirow{6}{*}{\begin{tabular}[c]{@{}l@{}} D3:UniBench\cite{zhangunibench}\\ (E-commerce)\end{tabular}} 
 & \multirow{6}{*}{\begin{tabular}[c]{@{}l@{}}Uniform\\  Tables:~1$m$ rows\\ JSON:~4$m$ nodes \\ XML:~1.4$m$ nodes\end{tabular}}   & Q16   & \multirow{6}{*}{\begin{tabular}[c]{@{}l@{}}R1(asin,orderID) R2(personID,lastname) \\ \$.[orderID,personID,orderline[asin]] \end{tabular}}     & OrderLine{[}asin{]}/price                      & $N^3$   &1.1$k$  \\\cline{3-3} \cline{5-7} 
                          &                                        & Q17        &                                                                                                                                                & T17=Invoice{[}orderID{]}/orderline{[}asin{]}/price & $N^3$   &<0.1$k$  \\\cline{3-3} \cline{5-7} 
                          &                                        & Q18        &                                                                                                                                                & T18=Invoice{[}orderID{]}//asin                     & $N^3$   &<0.1$k$   \\ \cline{3-3} \cline{5-7} 

 &   & Q19   &     & orderline/asin, orderline/price                      & $N^3$   &1.1$k$  \\\cline{3-3} \cline{5-7} 
                          &                                        & Q20        &                                                                                                                                                & 
                          Invoice(I)/orderID, I/orderline(O)/asin, I/O/price & $N^3$   &<0.1$k$  \\\cline{3-3} \cline{5-7} 
                          &                                        & Q21        &                                                                                                                                                & T17, T18                     & $N^3$   &<0.1$k$   \\ \hline

\multirow{3}{*}{\begin{tabular}[c]{@{}l@{}} D4:MIMIC-III\cite{dataset-johnson2016mimic}\\ (Clinical data)\end{tabular}} & \multirow{3}{*}{\begin{tabular}[c]{@{}l@{}}Uniform\\  Tables:0.5-10$m$ rows\\ JSON:~10$m$ nodes\end{tabular}}   & Q22   &  \begin{tabular}[c]{@{}l@{}} R1(RowID,ICUstayID,ItemID,CGID),\\ R2(RowID,SubjectID,ICUstayID,ItemID)\end{tabular}
    & T22=\$.[RowID,SubjectID,HADMID]                    & $N$   &<0.1$k$  \\\cline{3-4} \cline{5-7} 
                          &                                        & Q23        &    R1,R3(SubjectID,ItemID)                                                                                                                                            & 
                          T22 & $N^\frac{3}{2}$   &<0.1$k$  \\\cline{3-4} \cline{5-7} 
                          &                                        & Q24        &   R1,R2,R4(RowID,SubjectID,HADMID)                                                                                                                                             & \begin{tabular}[c]{@{}l@{}} T22, T23=\$.[RowID,ICUstayID,ItemID,CGID],\\ T24=\$.[RowID,SubjectID,HADMID,ICUstayID,ItemID,CGID]\end{tabular}                   & $N$   &<0.1$k$   \\ \bottomrule

\end{tabular}


%% file: section/acknowledgments.tex


%% file: section/appendix.tex
\section{Appendix}\label{sec:appendix} \label{sec:supplement_proof}

\begin{proof}[of \autoref{theo:np-hard}]

The proof idea is to polynomially reduce the \textit{1-IN-3SAT} problem~\cite{schaefer1978complexity} to our problem.
The goal is to show that once we find the size bound of the query, we achieve the solution for \textit{1-IN-3SAT}. We briefly exploit the \textit{1-in-3SAT} problem in this part:

\begin{itemize}[leftmargin=*, label=-]
\item Instance: A collection of clauses $l_1$, ..., $l_m$, $m >$ 1; each $l_i$
is a disjunction of exactly three literals.
\item Question: Is there a truth assignment to the variables occurring so that exactly one literal is true in each $l_i$?
\end{itemize}

\noindent \textbf{Query construction functions}:
\indent First, we introduce the basic functions for the proof as follows:
\begin{itemize}
    \item $\mathcal{K}_1((A,B),(C,D))$ constructs a tree pattern query $T=A[C]/D[B]$ and two relations $R_1(B,C)$ and $R_2(B,D)$, i.e., $\mathcal{K}_1=T\bowtie R_1\bowtie R_2$;
    \item $\mathcal{K}_2((A,B),(C,D,E,F)$ constructs a tree pattern query $T=D[A]/F/B//C/E$ and three relations $R_1(A,F)$, $R_2(A,C,E)$, and $R_3(B,C,E)$, i.e., $\mathcal{K}_2=T\bowtie R_1\bowtie R_2 \bowtie R_3$;
\end{itemize}

\noindent \textbf{Polynomial Construction} \indent Given any \textit{1-IN-3SAT} input $\phi$ with $m$ clauses $\mathcal{L}=\{l_1,...l_i...l_m \}$, where each $l_i$ is with $3$ variables, we create tree pattern and relation query set $\mathcal{Q}_{\phi}$, which is constructed as follows:

\noindent (1) \textbf{Variable construction:} In each clause $l_i$, for each variable in positive literal $x_j$, we establish two variables $x_{ij}$ and $a_{ij}$, while for its negative literal $\neg x_j$, we establish $y_{ij}$ and $b_{ij}$.

\noindent (2) \textbf{Within clause restriction:} For variable $x_{ik}$, $x_{ip}$, and $x_{id}$ in clause $l_i$, the construction form for each $l_i$ is following:

\begin{equation}\label{equ:np_encoding_within_clause}
    \begin{aligned}
    \mathcal{Q}_1 = \mathcal{K}_2((x_{ik},a_{ik}),(x_{ip},a_{ip},x_{id},a_{id}))\bowtie \\ \mathcal{K}_2((x_{ip},a_{ip}),(x_{ik},a_{ik},x_{id},a_{id}))\bowtie \\
    \mathcal{K}_2((x_{id},a_{id}),(x_{ik},a_{ik},x_{ip},a_{ip}))
    \end{aligned}
\end{equation}

\noindent (3) \textbf{Variable restriction:} For each variable's positive literal $x_{k}$ in $l_i$ and its negative literal $\neg x_{k}$ in $l_j$, the construction form is as follows:
\begin{equation}\label{equ:np_encoding_variable}
    \mathcal{Q}_2 = \mathcal{K}_1((x_{ik},a_{ik}),(y_{jk},b_{jk}))
\end{equation}

\noindent (4) \textbf{Between clauses restriction:} For each two clauses $l_i=\{x_{ik}, x_{ip}, x_{id}\}$ and $l_j=\{\neg x_{jk}, x_{jt}, x_{jc}\}$ with same variable $x_j$ but in form of both positive and negative literals, the restrictions are constructed as follows:
\begin{equation}\label{equ:np_encoding_between_clauses}
    \begin{aligned}
    \mathcal{Q}_3 = \mathcal{K}_2((x_{ip},a_{ip}), (x_{jt},a_{jt},x_{jc},a_{jc}))\bowtie \\ \mathcal{K}_2((x_{id},a_{id}), (x_{jt},a_{jt},x_{jc},a_{jc})) \bowtie \\ \mathcal{K}_2((x_{jt},a_{jt}), (x_{ip},a_{ip},x_{id},a_{id})) \bowtie \\ \mathcal{K}_2((x_{jc},a_{jc}), (x_{ip},a_{ip},x_{id},a_{id}))
    \end{aligned}
\end{equation}

\begin{table*}[t!]
\tiny
\centering
\caption{Example of reduction from a 1-IN-3SAT instance to a CMCQ.}
\label{tab:np_full_example}
\input{table/np_encoding}
\end{table*}



\noindent Therefore, $\mathcal{Q}_{\phi}=\mathcal{Q}_1 \bowtie, \mathcal{Q}_2 \bowtie \mathcal{Q}_3$.
For $\mathcal{Q}_{\phi}$ construction, we establish $2*m$ variables, which is $\mathcal{O}(m)$. For each clause in $\mathcal{Q}_1$, we call three times of $\mathcal{K}_2()$ function, with each of $14$ elements of variables. So the construction cost is $3*m*14$, which is $\mathcal{O}(m)$. For variable in $\mathcal{Q}_2$, we call one time of $\mathcal{K}_1()$ for each variable with both positive and negative literals. In worst case, each positive literal has a negative literals, then the construction cost is at most $8*(3*m/2)^2$, which is $\mathcal{O}(m^2)$. For each variable in $\mathcal{Q}_3$, the construction cost is linear (4 times, i.e., two variables to two variables) to  $\mathcal{Q}_2$, i.e., $4*8*(3*m/2)^2$, which is $\mathcal{O}(m^2)$. So the total construction cost is $\mathcal{O}(m^2)$, which is a polynomial transformation.
Note that we can replace relation tables by PC paths so that the CMCQ has only tree pattern queries.

Next, we introduce lemmas for the $\mathcal{K}_1()$ and $\mathcal{K}_2()$ functions as the proof foundation. 
\autoref{fig:np_building_block1} shows an example to illustrate the $\mathcal{K}_1()$ and \autoref{fig:np_building_block2} shows an example to illustrate the $\mathcal{K}_2()$ function and how it relates to $\phi$ instance.

\begin{figure}[h]
        \centering
        \includegraphics[width=0.85\linewidth]{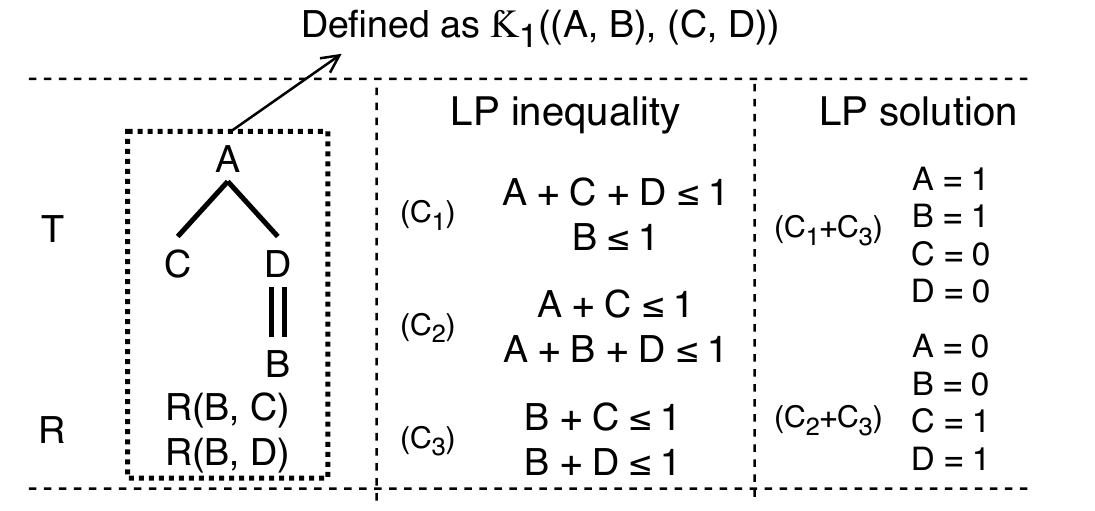}
        \caption{Illustration example for \autoref{lemma:np_building_block}}
        \label{fig:np_building_block1}
\end{figure}%

\begin{lemma}\label{lemma:np_building_block}
Given a tree pattern query $T=A[C]/D[B]$ and two relations $R_1(B,C)$ and $R_2(B,D)$, size bound $\Theta(N^2)$ is obtained for $|Q=T\bowtie R_1\bowtie R_2|$ by constructing either $|A|=|B|=\Theta(N)$ or $|A|=|B|=\mathcal{O}(1)$.
\end{lemma}

\begin{proof}
First, $T$ can be equivalently transformed into two alternative constraint sets $C_1=\{(A, C,D), (B)\}$ and $C_2=\{(A, C), (B,C,D)\}$, i.e., $C_1:\{ A+C+D\leq 1;~B\leq 1\}$ and $C_2:\{A+C\leq 1;~B+C+D\leq 1\}$ in LP form. And $R_1$ and $R_2$ is equivalent to $C_3=\{(B,C), (B,D)\}$, i.e., $C_3:\{B+C\leq 1;~B+D\leq 1\}$. The final size bound $\Theta(N^2)$ is obtained by either $C_1\cup C_3$ or $C_2\cup C_3$, with LP solution $A=B=1$ or $A=B=0$ ($C=D=1$), respectively. In $C_1\cup C_3$, by combining $A+C+D\leq1$ with $B+C\leq 1$ (or $B+D\leq 1$), we have $A+B+2C+D\leq2$ (or $A+B+C+2D\leq2$). If $C>0$ (or $D>0$), then $A+B+C+D<2$, meaning $Q$ can not achieve $\Theta(N^2)$. Thus $C=D=0$, meaning $A=B=1$.
Likewise in $C_2\cup C_3$, we obtain $A=B=0;~C=D=1$.
\end{proof}

\begin{figure}[h]
        \centering
        \includegraphics[width=0.85\linewidth]{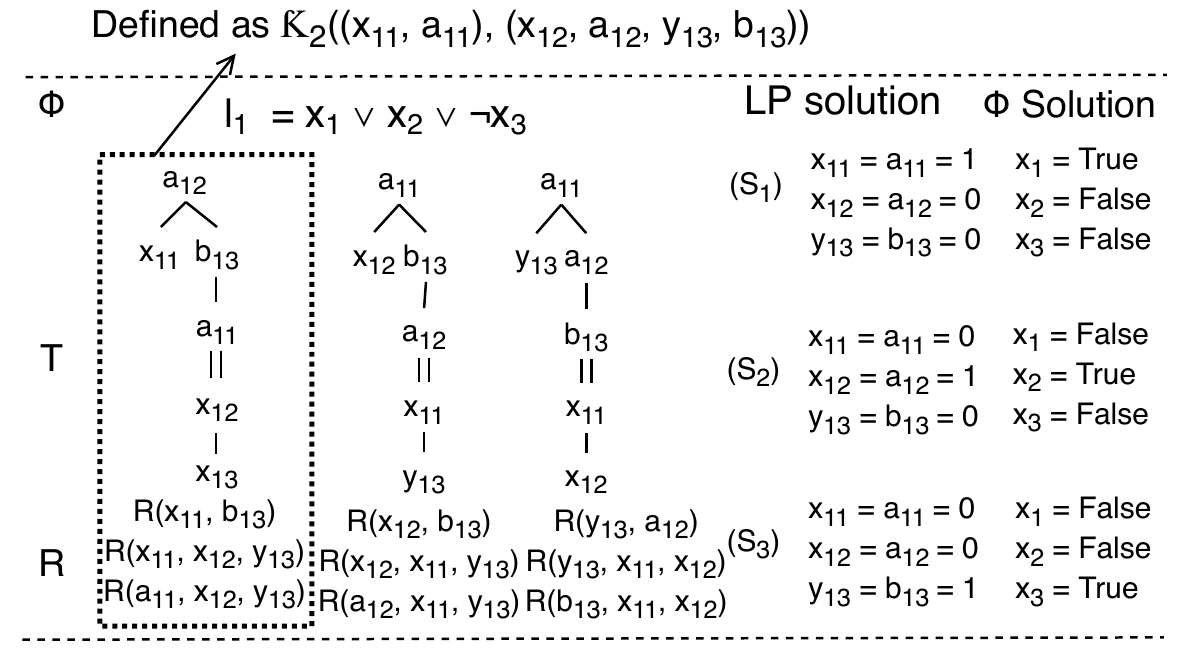}
        \caption{Illustration example for \autoref{lemma:np_building_block2} and the query construction for a single clause.}
        \label{fig:np_building_block2}
\end{figure}%

\begin{lemma}\label{lemma:np_building_block2}
Given a tree pattern query $T=D[A]/F/B//C/E$ and three relations $R_1(A,F)$, $R_2(A,C,E)$, and $R_3(B,C,E)$, size bound $\Theta(N^2)$ is obtained for $|Q=T\bowtie R_1\bowtie R_2 \bowtie R_3|$ by constructing either $|A|=|B|=\Theta(N)$ or $|A|=|B|=\mathcal{O}(1)$.
\end{lemma}

\begin{proof}
First, $T$ can be transformed into two alternative inequality sets i.e., $C_1:\{ A+B+D+F\leq 1;~C+E\leq 1\}$ and $C_2:\{A+D\leq 1;~B+C+D+E+F\leq 1\}$. And $R_1$, $R_2$, and $R_2$ is equivalent to $C_3:\{A+F\leq 1;~A+C+E\leq 1;~B+C+E\leq 1\}$. The final size bound $\Theta(N^2)$ is obtained by either $C_1\cup C_3$ or $C_2\cup C_3$, with LP solution $A=B=1$ or $A=B=0$, respectively. In $C_1\cup C_3$, by combining $ A+B+D+F\leq 1$ and $A+C+E\leq 1$, we have $2A+B+C+D+E+F\leq 2$, so $A$ should be $0$ to achieve final result $2$. By combining  $ A+B+D+F\leq 1$ and $B+C+E\leq 1$, we achieve $B=0$. Thus $A=B=0$ achieves $\mathcal{O}(N^2)$ in this case.  While in $C_2\cup C_3$, combining $B+C+D+E+F\leq 1$ with $A+D\leq 1$, $A+F\leq 1$, and $A+C+E\leq 1$, we obtain $D=0$, $F=0$ and $C=E=0$, respectively. Thus $A=B=1$ achieves $\Theta(N^2)$.
\end{proof}


\noindent \indent We now prove that $\mathcal{Q}_{\phi}$ achieves size bound $\Theta(N^{2m})$ if and only if
$\phi$ is satisfiable, where $m$ is number of clause.

\noindent (1) If we achieve $\Theta(N^{2m})$ size bound for $\mathcal{Q}_{\phi}$, then we can achieve $\Theta(N^{2})$ for each clause.
By within clause restriction $\mathcal{Q}_1$, for variable $x_{ij}$, $x_{ip}$, and $x_{id}$ in clause $l_i$, only one of $x_{ij}=a_{ij}$, $x_{ip}=a_{ip}$, and $x_{id}=a_{id}$ is assigned to be $1$ to achieve size bound $\Theta(N^{2})$, corresponding to only one of $x_j$, $x_p$, or $x_d$ in each clause $l_i$ to be assigned to be $True$.
By $\mathcal{Q}_2$ for each variable $x_k$, by \autoref{lemma:np_building_block}, either $x_{ik}=a_{ik}=1$ or $y_{jk}=b_{jk}=1$, corresponding to $x_i=True$ or $\neg x_i=True$ in $\phi$, respectively, meaning that $x_k$ and $\neg x_k$ can not be both $True$.
By $\mathcal{Q}_3$ for each two clauses $l_i=\{x_{ik}, x_{ip}, x_{id}\}$ and $l_j=\{\neg x_{jk}, x_{jt}, x_{jc}\}$ with common variable $x_k$ in different (one positive and one negative) literal forms, $\mathcal{Q}_3$ guarantees, if $x_{ik}=a_{ik}=0$, then $x_{ip}=a_{ip}=1$ or $x_{id}=a_{id}=1$, which forces $x_{jt}=n_{jt}=0$ and $x_{jc}=n_{jc}=0$, thus making $y_{jk}=b_{jk}=1$. Likewise, when $y_{jk}=b_{jk}=0$, we have $x_{ik}=a_{ik}=1$. $\mathcal{Q}_3$ corresponds to property that  $x_j$ and $\neg x_j$ can not be both $False$. By these assignments, each clause is with one $True$, and each variable pair $x_k$ and $\neg x_k$ can not be both $True$ and both $False$, then $\phi$ is satisfiable. 

\noindent (2) Conversely, if $\phi$ is satisfiable, then we have only one $True$ variable in one of  $x_j$, $x_p$, or $x_d$ in each clause $l_i$. We can assign LP values to corresponding variables $x_{ij}=a_{ij}$, $x_{ip}=a_{ip}$, or $x_{id}=a_{id}$ in $\mathcal{Q}_\phi$ such that each clause achieve $\Theta(N^{2})$. Since each variable pair $x_k$ and $\neg x_k$ can not be both $True$ and both $False$. So assigning LP values satisfies our $\mathcal{Q}_2$ and $\mathcal{Q}_3$ constructions. By summing up $m$ clause results, we achieve $|\mathcal{Q}_\phi|=\Theta(N^{2m})$. By our construction, as each clause has $6$ variables and is with $\mathcal{O}(N^{2})$ by \autoref{lemma:np_building_block2}, meaning that $\mathcal{Q}_\phi$ is upper bounded by $\mathcal{O}(N^{2m})$. Therefore, we achieve the size bound of $\mathcal{Q}_\phi$.

Therefore, $\mathcal{Q}_{\phi}$ achieves the size bound $\Theta(N^{2m})$ if and only if $\phi$ is satisfiable.
\end{proof}

\begin{example}
\autoref{tab:np_full_example} shows a full example to illustrate the reduction from a \textit{1-in-3SAT} instance. Given $\phi$, we construct the tree pattern query and relation by our polynomial transformation.

We can achieve the size bound $\Theta(N^{(2*m)})$ = $\Theta(N^6)$ by assigning $x_{11}=a_{11}=x_{22}=a_{22}=x_{34}=a_{34}=1$, where $m=3$ is the clause numbers. So the LP solution corresponds to the satisfiable solution for $\phi$, i.e., $x_1=True$, $x_2=False$, and $x_3=False$.

Conversely, satisfiable solution $x_1=True$, $x_2=False$, and $x_3=False$ in $\phi$ achieves size bound $\Theta(N^6)$ by assigning $x_{11}=a_{11}=x_{22}=a_{22}=x_{34}=a_{34}=1$.
\end{example}




%% file: table/np_encoding.tex
\setlength{\tabcolsep}{2.4pt}
\begin{tabular}{cccc>{\centering\arraybackslash}p{1.2in}}
\toprule
$\phi$ & $(x_1 \vee \neg x_2 \vee x_3)$  & $(\neg x_1 \vee x_2 \vee x_3)$  & $(\neg x_1 \vee x_4 \vee x_5)$ & Description\\\hline
$\mathcal{Q}_1$ & $\begin{aligned}
    \mathcal{K}_2((x_{11},a_{11}),(y_{12},b_{12},x_{13},a_{13})) \\ \mathcal{K}_2((y_{12},b_{12}),(x_{11},a_{11},x_{13},a_{13})) \\
    \mathcal{K}_2((x_{13},a_{13}),(x_{11},a_{11},y_{12},b_{12}))
    \end{aligned}  $
    & $\begin{aligned}
    \mathcal{K}_2((y_{21},b_{21}),(x_{22},a_{22},x_{23},a_{23})) \\ \mathcal{K}_2((x_{22},a_{22}),(y_{21},b_{21},x_{23},a_{23})) \\
    \mathcal{K}_2((x_{23},a_{23}),(y_{21},b_{21},x_{22},a_{22}))
    \end{aligned}  $
    & $\begin{aligned}
    \mathcal{K}_2((y_{31},b_{31}),(x_{34},a_{34},x_{35},a_{35})) \\ \mathcal{K}_2((x_{34},a_{34}),(y_{31},b_{31},x_{35},a_{35})) \\
    \mathcal{K}_2((x_{35},a_{35}),(y_{31},b_{31},x_{34},a_{34}))
    \end{aligned} $
    & One of $x_j$, $x_d$, $x_p$ in each clause $l_i$ is $True$ \\\hline
    
$\mathcal{Q}_2$ & $\begin{aligned}
    \mathcal{K}_1((x_{11},a_{11}),(y_{21},b_{21}) \\
    \mathcal{K}_1((x_{11},a_{11}),(y_{31},b_{31}) 
    \end{aligned}  $
    & $\begin{aligned}
    \mathcal{K}_1((x_{22},a_{22}),(y_{12},b_{12}) 
    \end{aligned}  $
    & 
    & $x_i$ and $\neg x_i$ cannot be both $True$ \\\hline

$\mathcal{Q}_3$ & $\begin{aligned}
    \mathcal{K}_2((y_{12},b_{12}),(x_{22},a_{22},x_{23},a_{23})) \\ \mathcal{K}_2((x_{13},a_{13}),(x_{22},a_{22},x_{23},a_{23})) \\
    \mathcal{K}_2((y_{12},b_{12}),(x_{34},a_{34},x_{35},a_{35})) \\ \mathcal{K}_2((x_{13},a_{13}),(x_{34},a_{34},x_{35},a_{35})) \\
    \end{aligned}  $
    & $\begin{aligned}
    \mathcal{K}_2((x_{22},a_{22}),(y_{12},b_{12},x_{13},a_{13})) \\ \mathcal{K}_2((x_{23},a_{23}),(y_{12},b_{12},x_{13},a_{13})) \\
    \mathcal{K}_2((y_{21},b_{21}),(x_{11},a_{11},x_{13},a_{13})) \\
    \mathcal{K}_2((x_{23},a_{23}),(x_{11},a_{11},x_{13},a_{13}))
    \end{aligned}  $
    & $\begin{aligned}
    \mathcal{K}_2((x_{34},a_{34}),(y_{12},b_{12},x_{13},a_{13})) \\ \mathcal{K}_2((x_{35},a_{35}),(y_{12},b_{12},x_{13},a_{13})) 
    \end{aligned}  $
    & $x_i$ and $\neg x_i$ cannot be both $False$ \\

\bottomrule
\end{tabular}